\documentclass[3p,number]{elsarticle}
\usepackage{etex}
\usepackage{amsmath}
\usepackage{amssymb}
\usepackage{amsthm}
\usepackage{array}
\usepackage{booktabs}
\usepackage{graphicx}
\usepackage{calc}
\usepackage{url}
\usepackage{multirow}
\usepackage{wrapfig}
\usepackage{tikz}
\usetikzlibrary{arrows,automata,shapes,shapes.multipart,decorations.markings,positioning}
\usepackage[all]{xy}
\usepackage[trim]{tokenizer} 
\newcommand{\outputs}{\mathsf{outputs}}
\newcommand{\inputs}{\mathsf{sorts}}
\newcommand{\cosub}{\sqsubseteq}

\newcommand{\sset}[2]{\left\{~#1  \left|
      \begin{array}{l}#2\end{array}
    \right.     \right\}}

\renewcommand{\vec}{\bar}
\newcommand{\bu}[1]{\underline{#1}}

\newcommand{\FF}{{\cal F}}
\newcommand{\ND}{\mathit{noneDown}}
\newcommand{\NU}{\mathit{noneUp}}
\newcommand{\OD}{\mathit{oneDown}}
\newcommand{\OU}{\mathit{oneUp}}
\newcommand{\once}{\mathit{once}}
\newcommand{\sorts}{\mathit{sorts}}
\newcommand{\update}{\mathit{update}}
\newcommand{\used}{\mathit{used}}
\newcommand{\valid}{\mathit{valid}}

\newcommand{\opt}{\mathit{opt}}

\newcounter{sncolumncounter}
\newcounter{snrowcounter}

\newtheorem{example}{Example}
\newtheorem{lemma}{Lemma}
\newtheorem{definition}{Definition}
\newtheorem{theorem}{Theorem}
\newtheorem{corollary}{Corollary}

\def \nodelabel#1{%
\setcounter{snrowcounter}{1}
 \foreach \i in {#1}{%
   \draw (\sncolwidth*\value{sncolumncounter},\value{snrowcounter}) node[anchor=south]{\i};
   \addtocounter{snrowcounter}{1}
 }
\addtocounter{snrowcounter}{-1}
 \addtocounter{sncolumncounter}{1}
}

\newcommand{\sncolwidth}{0.7} %

\def \addcomparator#1#2{%
    \draw (\sncolwidth*\value{sncolumncounter},#1) node[circle,fill=black,minimum size=4pt,inner sep=0pt,outer sep=0pt]{}--(\sncolwidth*\value{sncolumncounter},#2) node[circle,fill=black,minimum size=4pt,inner sep=0pt,outer sep=0pt]{};
}

\def \addlayer{%
  \addtocounter{sncolumncounter}{1}
}

\def \nextlayer{%
  \draw [dashed] (\sncolwidth*\value{sncolumncounter}+\sncolwidth,0.6)--(\sncolwidth*\value{sncolumncounter}+\sncolwidth,\value{snrowcounter}+0.6);
  \addtocounter{sncolumncounter}{2}
}

\def \nodeconnection#1{%
  \foreach \i in {#1}{%
    \GetTokens{nodesrc}{nodedest}{\i}
    \draw (\value{sncolumncounter}*0.7,\nodesrc) node[circle,inner sep=0pt,minimum size=3pt,fill=black]{}--(\value{sncolumncounter}*0.7,\nodedest) node[circle,inner sep=0pt,minimum size=3pt,fill=black]{};
  }
  \addtocounter{sncolumncounter}{1}
}
\newenvironment{sortingnetwork}[2]
{
  \setcounter{sncolumncounter}{0}
  \setcounter{snrowcounter}{#1}
  \def \sn@fullsize{15}
  \begin{tikzpicture}[scale=#2*0.7]
}
{
  \foreach \i in {1, ..., \value{snrowcounter}}
  {
    \draw (-\sncolwidth,\i)--(\sncolwidth*\value{sncolumncounter}+\sncolwidth,\i);
  }
  \end{tikzpicture}
}
\makeatother

\begin{document} 

\title{Sorting Networks: to the End and Back Again\tnoteref{t1}} 

\tnotetext[t1]{Supported by the Israel Science Foundation, grant
  182/13, by a fellowship within the FITweltweit programme of the German Academic Exchange Service (DAAD), and by the Danish Council for Independent Research, Natural
  Sciences. Computational resources provided by an IBM Shared
  University Award (BGU) and the Dependable Systems Group, University of Kiel.}  

\author[bgu]{Michael Codish} \ead{mcodish@cs.bgu.ac.il} 
\author[sdu]{Lu\'\i s Cruz-Filipe\corref{cor1}} \ead{lcf@imada.sdu.dk} 
\author[uk]{Thorsten Ehlers} \ead{the@informatik.uni-kiel.de} 
\author[uk]{Mike M\"uller} \ead{mimu@informatik.uni-kiel.de} 
\author[sdu]{Peter Schneider-Kamp} \ead{petersk@imada.sdu.dk}
\cortext[cor1]{Corresponding author; tel.\ +45 6550 2387, fax +45 6550 2373}

\address[bgu]{Department of Computer Science, Ben-Gurion University of
  the Negev, Israel} 
\address[sdu]{Department of Mathematics and
  Computer Science, University of Southern Denmark}
\address[uk]{Institut f\"ur Informatik,
  Christian-Albrechts-Universit\"at zu Kiel, Germany}

\begin{abstract}
  This paper studies new properties of the front and back ends of a
  sorting network, and illustrates the utility of these in the search
  for new bounds on optimal sorting networks.
  Search focuses first on the ``out-sides'' of the network and then on
  the inner part.
  All previous works focus only on properties of the front end of
  networks and on how to apply these to break symmetries in the
  search.
  The new, out-side-in, properties help shed understanding on how
  sorting networks sort, and facilitate the computation of new bounds
  on optimal sorting networks.
  We present new parallel sorting networks for 17 to 20 inputs.  For
  17, 19, and 20 inputs these networks are faster than the previously
  known best networks. 
  For 17 inputs, the new sorting network is shown optimal in the sense
  that no sorting network using less layers exists.
\end{abstract}

\maketitle

\section{Introduction}
\label{sec:intro}
Sorting networks are a classical model for sorting algorithms on
fixed-length lists. The elements to sort are placed on the input
channels of a network of compare-and-exchange units connecting
pairs of channels. The sequence of these units, called comparators,
uniquely determines the algorithm. Consecutive comparators can be
viewed as a ``parallel layer'' if no two touch on the same channel.
Ever since sorting networks were introduced, there has been a quest to
find optimal sorting networks for particular small numbers of inputs:
optimal depth networks (in the number of parallel layers), as well as
optimal size (in the number of comparators).
For an overview on sorting networks see, for example,
Knuth~\cite{Knuth73} or Parberry~\cite{Parberry87}.

Sorting networks are data-oblivious sorting algorithms, i.e., the
sequence of comparisons they perform is independent of the input
list. Thus, they are well suited to implementation as hardware
circuits, where their size determines the number of gates used and
their depth the delay of the circuit. Recent
work~\cite{ourhopefulpaper,Furtak2007} has also shown their
performance potential in software implementations of sorting as base
cases of general recursive sorting algorithms.

The formal simplicity of the model, the ubiquitousness of sorting, and
the computationally challenging nature of the optimization problems
even for very small instances is a potent mixture, which  has
intrigued computer scientists since the middle 1950s. 
In the late 1960s, Donald E. Knuth and Robert W. Floyd authored a
series of articles on the topic, essentially settling all optimality
questions for networks with up to $8$ channels. The results of this
effort are presented in the chapter dedicated to sorting networks in
the famous monograph by Knuth~\cite{Knuth73}.
Nevertheless, no further breakthroughs were made before Ian Parberry
in 1989~\cite{Parberry89,Parberry91} settled the depth-optimality for $9$ and
$10$ channels, using a then state-of-the-art Cray supercomputer. Only
very recently, depth optimality for $11$ to $16$~\cite{BZ14} and
size optimality for $9$ and $10$~\cite{ourICTAIpaper} were finally
settled.

The new results are due to both an exponential increase in
computational power, vastly improved search methods, and new
theoretical results on symmetries. To illustrate the first point, the
program from \cite{Parberry89,Parberry91} for $9$ channels runs in just $12.3$
seconds on a single thread of Parberry's current desktop computer
instead of $200$ hours on a supercomputer~\cite{Parberry15}.
Processing the case of $11$ channels requires more than $2$ months,
and that of $13$ channels is infeasible.
 In contrast, using SAT solvers and
stronger notions of symmetry, the results for $11$ and $13$ channels
were obtained in approx.\ $15$ minutes and $11$ hours, respectively
\cite{BCCSZ14}.

This paper reveals new properties of sorting networks that deepen
our understanding of how they work and facilitate our search for new
results on depth optimality for sorting networks with more than $16$
channels.
Until now,
all relevant progress on depth optimality has been made by
identifying and breaking symmetries in the first layer(s) of the
sorting networks, usually based on the %
property of sorting networks being closed under a particular variant
of permutation.
For example, Parberry's results on depth optimality~\cite{Parberry89,Parberry91}
derive from the observation that the first layer can be fixed, while
the results of Bundala and Z\'avodn\'y~\cite{BZ14} derive by fixing
the first two layers and using an improved notion of symmetry.
All of the recent results on optimality of sorting networks are
obtained as encodings to Boolean satisfiability, where a SAT solver is
applied to determine the existence of a sorting network of a given
size or depth. In this approach, the size of the SAT encodings is
polynomial in the set of all inputs to the sorting network, and hence
exponential in the number of elements to sort. This complexity seems
inherent. Fixing some layers of the network leads to significant
reduction in the size of the encoding, and is key to these new results.

The first contribution of this paper is a study of the \emph{end}
layers of a sorting network.  We show that the comparators in the last
layer of a sorting network are of a very particular form, and that
those in penultimate layer are also of limited structure. We show that
these results translate to constraints on the search for optimal
sorting networks and lead to significant performance improvements when
searching for networks of a given depth using SAT encodings. 
We also obtain theoretical results that, while not directly impacting
the SAT-solving times, establish a duality between the symmetries in
the first layers and in the last layers of a sorting network.

The second contribution is an improvement of the SAT encoding used by
\cite{BZ14}, again reducing the size of the SAT problems and improving
SAT solving performance. Together with the first contribution, we are
able to find a sorting network on $17$ channels with depth $10$ very
efficiently, by fixing the first $3$ layers as in a Green
filter~\cite{Coles12}. This result improves the previous upper bound
of $11$ to $10$.

While the upper bound can simply be improved by finding a smaller
sorting network, improving the lower bound from $9$ to $10$ requires
showing that there does not exist a sorting network of depth $9$ on
$17$ channels.  In other words, this requires reasoning over the full
space of $9$-layer sorting networks.
Fortunately, extending on the idea of \cite{BZ14}, we do not need to
consider all $2^n$ inputs: if there is no $9$-layer network that sorts
a subset of the inputs, then there is no $9$-layer sorting network.

The third contribution of this paper is the observation that the
performance of the SAT-based search for extensions of a prefix using
the above idea heavily depends on the prefix chosen. We show how to
refine prefixes by permutation in such a way that the size of the SAT
problems is vastly reduced, and that this translates to real-world
performance gains.

By combining all three contributions, the process of SAT solving for
unsatisfiable instances experiences a speed-up of several orders of
magnitude. This allows us to show that there is indeed no depth-$9$
sorting network on $17$ channels, i.e., that the new depth-$10$
sorting network we found is not only faster, but that it is also
optimal.
This paper is structured as follows.  In Section~\ref{sec:backgr} we
summarize relevant concepts, theory, and results on sorting networks
and present an improved SAT encoding.  Section~\ref{sec:last} presents
the new theoretical results about the structure of the last layers and
exploits them to improve the SAT encoding.
Section~\ref{sec:permuting} then focuses on permuting prefixes in
order to optimize the SAT encoding.  As a result, we present in
Section~\ref{sec:upbound} more efficient sorting networks than
previously known, and in Section~\ref{sec:lowbound} we prove a new
lower bound for $17$ channels.  We summarize our results in
Section~\ref{sec:concl}.

This paper combines and significantly extends preliminary work
presented in~\cite{ourLATApaper} and in~\cite{DBLP:conf/cie/EhlersM15}.

\section{An Overview of Sorting Networks}
\label{sec:backgr}

This section first introduces background material on sorting networks,
together with an overview of previous results that are used in the
remainder of the paper. At its end we present an improved SAT encoding
for finding sorting networks.

\subsection{Sorting Network Classics}

A \emph{comparator network} on $n$ channels is a sequence
$C=L_1,\ldots,L_d$ of \emph{layers}, where each layer $L_k$ is a
non-empty set of pairs $(i,j)$, with $1\leq i<j\leq n$, such that: if
$(i,j),(i',j')\in L_k$, then $\{i,j\}\cap\{i',j'\}=\emptyset$.  The
pairs $(i,j)$ are called \emph{comparators}; the \emph{depth} of $C$
is $d$, and the \emph{size} of $C$ is the total number of comparators
in all its layers, $\sum_{k=1}^d|L_k|$.

An $n$-channel comparator network is viewed as a function with $n$
inputs, which propagate through the network to give $n$ outputs. Each
comparator potentially changes the order between values on its two
channels. 
Let $\vec x$ be a sequence of inputs from some totally ordered set and
denote by $\vec x_\ell$ its $\ell$-th element.  Sequence $\vec x$
propagates through $C$ as follows: at each layer $L_k$, if $(i,j)\in
L_k$ and $\vec x_i>\vec x_j$, then $\vec x_i$ and $\vec x_j$ are
interchanged.  More formally, we define a sequence $\vec
x^0,\ldots,\vec x^d$ as follows:
\begin{equation}
\begin{aligned}
  \vec x^0 &= \vec x \\
  \vec x^k_i &=
    \begin{cases}
      \vec x^k_j & (i,j)\in L_k, \vec x^{k-1}_i>\vec x^{k-1}_j \\
      \vec x^k_j & (j,i)\in L_k, \vec x^{k-1}_j>\vec x^{k-1}_i \\
      \vec x^k_i & \mbox{ otherwise}
    \end{cases}
\end{aligned}
\label{eq:snRun}
\end{equation}

The \emph{output} of $C$ on $\vec x$, which we denote by $C(\vec x)$, is
$\vec x^d$.
The set of all outputs of $C$ is $\outputs(C)$, and $C$ is said to be a
\emph{sorting network} (on $n$ channels) if all these sequences are sorted.
The zero-one principle (see e.g.~\cite{Knuth73}) states that a
comparator network $C$ is a sorting network if and only if $C$ sorts
all Boolean inputs. Hence, in the remainder of the paper, we consider
only comparator networks with Boolean inputs.

\begin{figure}[htb]
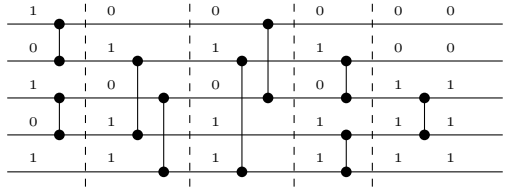

 \centering
 \begin{sortingnetwork}{5}{0.7}
   \nodelabel{\tiny 1,\tiny 0,\tiny 1,\tiny 0,\tiny 1}
   \addcomparator45
   \addcomparator23
   \nextlayer
   \nodelabel{\tiny 1,\tiny 1,\tiny 0,\tiny 1,\tiny 0}
   \addcomparator24
   \addlayer
   \addcomparator13
   \nextlayer
   \nodelabel{\tiny 1,\tiny 1,\tiny 0,\tiny 1,\tiny 0}
   \addcomparator14
   \addlayer
   \addcomparator35
   \nextlayer
   \nodelabel{\tiny 1,\tiny 1,\tiny 0,\tiny 1,\tiny 0}
   \addcomparator12
   \addcomparator34
   \nextlayer
   \nodelabel{\tiny 1,\tiny 1,\tiny 1,\tiny 0,\tiny 0}
   \addcomparator23
   \addlayer
   \nodelabel{\tiny 1,\tiny 1,\tiny 1,\tiny 0,\tiny 0}
 \end{sortingnetwork}
 \caption{The sorting network
   $\{(1,2),(3,4)\},\{(2,4),(3,5)\},\{(1,3),(2,5)\},\{(2,3),(4,5)\},\{(3,4)\}$,
   on $5$ channels, operating on the input $10101$.}
 \label{fig:sort-5}
\end{figure}

It is common practice to represent comparator networks graphically as a Knuth diagram,
as showed in Figure~\ref{fig:sort-5}.   Channels are depicted as horizontal
lines, with values traveling from left to right, and comparators as
vertical lines connecting two channels.  Layers are sometimes made
explicit and then separated by dashed vertical lines.

For sorting networks, two natural optimization questions arise.
\begin{itemize}
\item The \emph{optimal size problem}: what is the minimum number of
  comparators, $s_n$, in a sorting network on $n$ channels?
\item The \emph{optimal depth problem}: what is the minimum number of
  layers, $t_n$, in a sorting network on $n$ channels?
\end{itemize}

The sorting network in Figure~\ref{fig:sort-5}, constructed in $5$
layers with a total of 9 comparators, is both size-optimal and
depth-optimal. In general, there is not always a single network that
is optimal for both criteria.
Table~\ref{tab:optimal} presents the currently known values of $s_n$ and
$t_n$. In case the precise value is unknown, the table presents the
best known lower and upper bounds. For instance, the smallest open
problem for the optimal size problem is for $11$ channels. The minimum
number of comparators in a sorting network on $11$ channels is known
to be between $33$ and $35$, but the precise number is not known.
Values which are derived as contributions of this paper are set in
boldface.
The values of $s_1$--$s_8$ and $t_1$--$t_8$ were already known in the
1960s and are included in~\cite{Knuth73}, as well as all upper bounds
up to $s_{16}$ and $t_{16}$.  The remaining upper bounds for $s_n$
can be found in~\cite{ValsalamM13}, and the upper bounds
$t_{17},t_{18}\leq 11$ and $t_{19},t_{20}\leq 12$ in \cite{Batcher2011}.
A computer-assisted proof of the exact values of $s_9$ and $s_{10}$
was first presented in 2014~\cite{ourICTAIpaper}, and the remaining
lower bounds for $s_n$ follow from these by the result of van
Voorhis~\cite{Voorhis72} that $s_n\geq s_{n-1}+\lceil\lg n\rceil$.
The exact values of $t_9$ and $t_{10}$ were first shown by Parberry in
1989~\cite{Parberry89,Parberry91}, and those of $t_{11}$--$t_{16}$ by Bundala and
Z\'avodn\'y in 2014~\cite{BCCSZ14,BZ14}.
The latter results imply that $t_n\geq 9$ for every $n\geq 17$.
Among the contributions of this paper are the proofs
that $t_{17}=10$ and $t_{20}\leq 11$. As a consequence, we obtain that
$10\leq t_n\leq 11$ for $18\leq n\leq20$.%

\begin{table}
  \centering
  \begin{tabular}{c|cccccccccccccccccccc}
    \toprule
    $n$ & 1 & 2 & 3 & 4 & 5 & 6 & 7 & 8 & 9 & 10 & 11 & 12 & 13 & 14 & 15 & 16 &
    17 & 18 & 19 & 20
    \\ \midrule
    \multirow2*{$s_n$} & \multirow2*0 & \multirow2*1 & \multirow2*3 &
    \multirow2*5 & \multirow2*9 & \multirow2*{12} & \multirow2*{16} &
    \multirow2*{19} & \multirow2*{25} & \multirow2*{29} & 35 & 39 & 45 & 51 & 56
    & 60 & 71 & 78 & 86 & 92 \\
    &&&&&&&&&&& 33 & 37 & 41 & 45 & 49 & 53 & 58 & 63 & 68 & 73 \\ \midrule
    \multirow2*{$t_n$} & \multirow2*0 & \multirow2*1 & \multirow2*3 &
    \multirow2*3 & \multirow2*5 & \multirow2*5 & \multirow2*6 & \multirow2*6 &
    \multirow2*7 & \multirow2*7 & \multirow2*8 & \multirow2*8 & \multirow2*9 &
    \multirow2*9 & \multirow2*9 & \multirow2*9 & \multirow2*{\bf 10} &
    11 &\textbf{11} & \textbf{11} \\
    &&&&&&&&&&&&&&&&&  & \textbf{10} & \textbf{10} & \textbf{10}\\ \bottomrule
  \end{tabular}
  \caption{Best known values and bounds on optimal size ($s_n$) and depth
    ($t_n$) of sorting networks on $n$ inputs, for $n\leq
    20$. The contributions of this paper are shown in boldface.}
  \label{tab:optimal}
\end{table}

For $\vec x, \vec y \in\{0,1\}^n$ we write $\vec
x\leq\vec y$ to denote that every bit of $\vec x$ is less than or
equal to the corresponding bit of $\vec y$, and $\vec x<\vec y$ for
$\vec x\leq\vec y$ and $x\neq y$.
The following two observations will be be instrumental for proofs in
later sections.  
\begin{lemma}%
\label{lem:sorted}
Let $C$ be a comparator network with $d$ layers, $\vec x^0 \in
\{0,1\}^n$, and $\vec x^0,\ldots,\vec x^d$ the sequence defined in
Equation~\eqref{eq:snRun}.
If $\vec x^0$ has $r$ leading zeroes and $s$ trailing ones ($r+s\leq
n$) then also each sequence $\vec x^i$, with $1\leq i\leq d$, has $r$
leading zeroes and $s$ trailing ones. In particular, if $\vec x^0$ is
sorted, then the sequence $\vec x^0,\ldots,\vec x^d$ is constant.
\end{lemma}
To the best of our knowledge, Lemma~\ref{lem:sorted} has never been
stated explicitly, but it is used implicitly in~\cite{BZ14}.  The
observation that $\vec x^0=\vec x^1=\cdots=\vec x^d$ when $\vec x^0$
is sorted was already made in~\cite{Knuth73}.

\begin{lemma}[Theorem~4.1 in~\cite{Batcher2011}]
  \label{lem:monot}
  Let $C$ be a comparator network and $\vec x,\vec y\in\{0,1\}^n$ be such that
  $\vec x\leq\vec y$. Then $C(\vec x)\leq C(\vec y)$.
\end{lemma}

\subsection{On Fixing the First Two Layers}

Throughout this paper we focus mostly on the optimal depth problem.
We aim to find new results to strengthen and extend the type of
techniques applied in previous works that helped establish the values of
$t_9$--$t_{16}$.

In order to establish that $t_n>k$, it is necessary to show that no
comparator network on $n$ channels with depth $k$ can sort all inputs.
Except in very simple cases, all known techniques require the analysis
of the entire search space of depth-$k$ comparator networks.  
However, this space grows very rapidly: the number of all possible
layers on $n$ channels coincides with the number of matchings in a
complete graph on $n$ vertices, which is exponential in~$n$, so
exhaustively generating it only works for small values of $n$.

Previous work reduces the search space by considering symmetries in
the first two layers of a sorting network.
The following terminology is useful.  Given two comparator networks
$C_1=L_1,\ldots,L_d$ and $C_2=L'_1,\ldots,L'_{d'}$, their
concatenation is the network
$C=C_1;C_2=L_1,\ldots,L_d,L'_1,\ldots,L'_{d'}$, and we say that $C_1$
is a \emph{$d$-layer prefix} (or simply a \emph{prefix}) of $C$.  A
layer is \emph{maximal} if it contains exactly $\lfloor\frac
n2\rfloor$ comparators.
We will call a set $\mathcal F$ of comparator networks of depth $k$ a
\emph{complete set of filters} (of depth $k$) if there is an optimal-depth
sorting network on $n$ channels of the form $F;C$ with $F\in\mathcal F$.
Parberry proved~\cite{Parberry89,Parberry91} that $\{L\}$ is a complete set of
filters of depth~$1$ for any maximal layer $L$ on $n$ channels,
and illustrated results using the particular first layer filter
which we denote 
\begin{equation}
  \label{eq:filterParberry}
  L_1^P = \sset{ ( 2i-1, 2i)}{1 \leq i \leq 
                 \left\lfloor \frac{n}{2} \right\rfloor }
\end{equation}
Parberry also considered restrictions on the second layer of a sorting
network, noting that one need not consider  prefixes that
are identical modulo permutations of channels that leave the first
layer fixed.

Parberry's ideas were later extended by Bundala and
Z\'avodn\'y~\cite{BZ14}, who pruned the search space by further
restricting the possibilities for the second layer of a sorting
network.  These results were strengthened in~\cite{ourSYNASCpaper},
which presents an efficient algorithm to generate complete sets of
filters $R_n$ of depth~$2$ for any~$n$.  The sizes of these sets for
$n\leq 20$ are given in Table~\ref{tab:Rn}; for comparison, we also
detail the order of magnitude of the total number $|G_n|$ of possible
layers.
The construction of the sets $R_n$ is described in
\cite{BZ14,ourSYNASCpaper} and uses two different ideas.  The first is
the notion of \emph{saturation}: a two-layer network $L_1;L_2$ is
saturated if $L_1$ is maximal and it is not possible to find
$L_2'\supsetneq L_2$ such that
$\outputs(L_1;L_2')\subsetneq\outputs(L_1;L_2)$.  It follows that the
set of all saturated networks is a complete set of two-layer filters.
The second notion involves permutations.  Given two comparator
networks $C$ and $C'$ on $n$ channels, we say that $C$ \emph{subsumes}
$C'$, denoted $C\preceq C'$, if
$\pi(\outputs(C))\subseteq\outputs(C')$ for some permutation $\pi$ of
$\{1,\ldots,n\}$.  In this situation, if $C'$ can be extended to a
sorting network, then so can $C$ (within the same size or depth),
which allows $C'$ to be removed from the search space.
In fact, the sets $R_n$ are not unique, as replacing any network in a
complete set of filters by one of its permutations yields another
complete set of filters. In the
following, we denote by $\FF_P$ a complete set of filters that all have
the layer $L_1^P$ of Equation~(\ref{eq:filterParberry}) as first layer.

\begin{table}\small
  \[\begin{array}{c|cccccccccccccccccc}
  \toprule
  n & 3 & 4 & 5 & 6 & 7 & 8 & 9 & 10 & 11 & 12 & 13 & 14 & 15 & 16 & 17 & 18 &
  19 & 20 \\ \midrule
  |G_n| & 10^0 & 10^1 & 10^1 & 10^1 & 10^2 & 10^2 & 10^3 & 10^3 & 10^4 & 10^5 & 10^5 & 10^6 & 10^7 & 10^7 & 10^8 & 10^8 & 10^9 & 10^{10} \\
  |R_n| & 1 & 2 & 4 & 5 & 8 & 12 & 22 & 21 & 48 & 50 & 117 & 94 & 262 & 211 &
  609 & 411 & 1{,}367 & 894\\ \bottomrule
  \end{array}\]

  \caption{Order of magnitude of $|G_n|$, the number of possible second layers
    on $n$ channels, and value of $|R_n|$, a complete set of two-layer filters,
    for $n\leq 20$.}
  \label{tab:Rn}
\end{table}

Given a complete set $\FF$ of (two layer) filters on $n$ channels, one
may search for a depth-$d$ sorting network considering (in parallel)
the independent search problems of extending each $F\in\FF$ to a sorting network
of depth $d$. It is sufficient to solve any one of these
separate search problems. Similarly, to prove that no depth-$d$
sorting network on $n$ channels exists, it is sufficient to show (in
parallel) that none of the filters $F\in\FF$ extend to a depth-$d$
sorting network.

\subsection{Prefixes with More than Two Layers}

Over the years, when searching for smaller sorting networks (improved
upper bounds) it has become common practice to fix the prefix of the
network, typically with more than just two layers. 
For example, in 1969, Green found the smallest network on 16 channels
known until today (with $60$ comparators) applying what is now called
a \emph{Green filter} \cite{Knuth73}. A Green filter on $n$ channels
(where $n$ is a power of $2$) consists of $\log_2 n$ maximal layers,
structured as exemplified for 16 channels in Figure~\ref{fig:green}.
It is also interesting to note that Parberry's construction for
Pairwise sorting networks introduced in \cite{Parberry92} includes a
Green filter (by construction), although it is not presented as
such. Codish and Zazon showed in \cite{CodishZ10} that Parberry's
Pairwise network is also a simple reordering of the layers in Batcher's
Odd-Even network~\cite{Batcher68}.

\begin{figure}[htb]
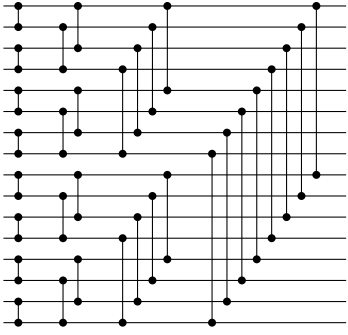

    \centering
    \begin{sortingnetwork}{16}{0.4}
        \nodeconnection{ {1,2}, {3,4}, {5,6}, {7,8}, {9,10}, {11,12}, {13,14}, {15,16}}
        \addtocounter{sncolumncounter}{2}
        \nodeconnection{ {1,3}, {5,7}, {9,11}, {13,15}}
        \nodeconnection{ {2,4}, {6,8}, {10,12}, {14,16}}
        \addtocounter{sncolumncounter}{2}
        \nodeconnection{ {1,5}, {9,13}}
        \nodeconnection{ {2,6}, {10,14}}
        \nodeconnection{ {3,7}, {11,15}}
        \nodeconnection{ {4,8}, {12,16}}
        \addtocounter{sncolumncounter}{2}
        \nodeconnection{ {1,9}}
        \nodeconnection{ {2,10}}
        \nodeconnection{ {3,11}}
        \nodeconnection{ {4,12}}
        \nodeconnection{ {5,13}}
        \nodeconnection{ {6,14}}
        \nodeconnection{ {7,15}}
        \nodeconnection{ {8,16}}
    \end{sortingnetwork}
    \caption{A Green filter on $16$ channels.}
    \label{fig:green}
\end{figure}

When searching for improved upper bounds on the depth (or size) of a
sorting network, a common criterion to select a particular filter is
to consider the number of its (unsorted) outputs.
For example, Coles~\cite{Coles12} suggests to evaluate filters in this way,
introducing the notion of Square filters, and observing that, for
$16$-channel networks, these produce fewer outputs than Green filters.
Later, Baddar and Batcher~\cite{Batcher2011} developed the
\texttt{sortnet} program to facilitate the specification of a filter
and its evaluation with respect to its number of outputs.

In this paper we show that, using the first three layers of the Green
filter depicted in Figure~\ref{fig:green} together with our other
contributions, we can find a $10$-layer sorting network
on $17$ channels. This improves on the previously smallest-known depth of a
sorting network on $17$ channels ($11$ layers).

\subsection{SAT Encoding for Depth-Restricted Sorting Networks}
\label{backgr_sat}

A first approach to encode sorting networks as formulae in
propositional logic was suggested by Morgenstern and
Schneider~\cite{MorgensternS11}.  However, their encoding to SAT did not
prove sufficient to find new results concerning optimal-depth
networks, as it did not scale for $n > 10$.
Bundala and Z\'avodn\'y~\cite{BZ14} continued this approach, and
introduced a better SAT encoding that was able both to find sorting
networks of optimal depth with up to $13$ channels and to prove their
optimality, implying the optimal depth of the best known networks with
up to $16$ channels.
This
was the first SAT formulation that led to new results on optimal-depth
sorting networks.
Bundala {\it et al.}~\cite{BCCSZ14} further improve this encoding to
establish optimal depth for networks with $n\leq 16$ channels directly.
In this paper, we obtain results applying further extensions and
optimizations of the SAT encoding presented in~\cite{BCCSZ14}, hence
for sake of completeness, we recall it here.

A comparator network of depth $d$ on $n$ channels is represented by a
set of Boolean variables $C^d_n=\sset{g^k_{i,j}}{1 \leq i < j \leq n,1
  \leq k \leq d}$, the value of $g^k_{i,j}$ indicating whether there is a
comparator on channels $i$ and $j$ in layer $k$ in the network
or not.
Furthermore, denoting the inputs into the network by the
variables~$v^0_i$, with $1 \leq i \leq n$, the variables $v^k_i$, with $1
\leq i \leq n$ and $1 \leq k \leq d$, store the value on channel $i$
in the network after layer $k$.  
The existence of a sorting network of depth $d$ on $n$ channels is
then encoded in terms of the following propositional constraints: 
\begin{align*}
 \once^k_i(C^d_n) = & \bigwedge_{1 \leq i \neq j \neq \ell \leq n} 
       \left( \neg g^k_{\min(i, j), \max(i, j) } \vee 
              \neg g^k_{\min(i, \ell), \max(i, \ell) } \right) \\
 \valid(C^d_n) = & \bigwedge_{1 \leq k \leq d, 1 \leq i \leq n} \once^k_i(C^d_n) \\
 \used^k_i(C^d_n) = & \bigvee_{j <i} g^k_{j, i} \vee \bigvee_{i < j} g^k_{i, j} \\
 \update^k_i(C^d_n,\bar v,w) = 
      & \left( \neg \used_i^k(C^d_n) \rightarrow
          (w \leftrightarrow v_i) \right) \wedge \\
      &\bigwedge_{1 \leq j < i} \left( g^k_{j,i} \rightarrow 
          \left(w \leftrightarrow (v_j \vee v_i)\right) \right) \wedge \\
      &\bigwedge_{i < j \leq n} \left( g^k_{i,j} \rightarrow 
           \left(w \leftrightarrow (v_j \wedge v_i )\right) \right)
\end{align*}
Here, $\once$ encodes the fact that each channel may be used only once
in one layer, and $\valid$ enforces this constraint for each channel
and each layer.
The $\update$ constraint specifies the value $w$ in channel $i$ after layer
$k$, given that the values before that layer are those in $\bar v=(v_1,\ldots,v_n)$.
In the next constraint, $\bar x = (x_1,\ldots,x_n) \in \{0,1\}^n$ and 
$\bar y = (y_1,\ldots ,y_n)$ is the sequence obtained by sorting $\bar
x$, while $\vec v^k=(v^k_1,\ldots,v^k_n)$ are fresh variables for $k=0,\dots,d$.
\begin{align}\label{eq:sorts}
  \sorts(C^d_n, \bar x) = 
    \bigwedge_{1 \leq i \leq n} (v^0_i \leftrightarrow x_i) \wedge 
    \bigwedge_{1 \leq k \leq d, 1 \leq i \leq n} 
          \update^k_i(C^d_n,\vec v^{k-1},v^k_i) \wedge 
          \bigwedge_{1 \leq i \leq n} (v^d_i \leftrightarrow y_i) 
\end{align}
The constraint $\sorts$ encodes whether a certain input, $\bar x$, is sorted
by the network $C^d_n$.  For this purpose, the values after layer $d$
(i.e., the outputs of the network) are compared to the vector $\bar y$.
A sorting network for $n$ channels on $d$ layers exists if and only if the following constraint
is satisfiable.
\begin{equation}\label{eq:sat}
  \varphi(n,d) = \valid(C^d_n) \wedge \bigwedge_{\bar x \in \{0,1\}^n} \sorts(C^d_n, \bar x)
\end{equation}
When seeking a sorting network of depth $d$ that extends a given
prefix $P$, we consider instead the proposition {\it``There is a
  comparator network on $n$ channels of depth $d-|P|$ that sorts
  $\outputs(P)$''}, or, formally: 
\begin{equation}\label{eq:satP}
  \varphi_P(n,d) = \valid(C^{d-|P|}_n) \wedge \bigwedge_{\bar x \in \outputs(P)} \sorts(C^{d-|P|}_n, \bar x)
\end{equation}

\subsection{Our Improved SAT Formulation}
\label{sec:improved}
In this paper we further optimize the SAT encoding
detailed in Section~\ref{backgr_sat}. We consider three
optimizations.
The first two optimizations decrease the size of each of the
(exponentially many) conjuncts in the formulae of
Equations~(\ref{eq:sat}) and~(\ref{eq:satP}).
First, based on Lemma~\ref{lem:sorted}, per input sequence $\bar x^0$,
we know that a leading zero (trailing one) in $\bar x^0$ is also a
leading zero (trailing one) in $\bar x^k$ at each layer $k$. So, we
fix these values at encoding time.
This idea was already used by Bundala and Z\'avodn\'y~\cite{BZ14}.
Second, %
for each input sequence $\bar x^0$, we know that for all pairs of
channels $i<j$, if $x^0_i$ is a leading zero or $x^0_j$ is a trailing
one, %
then the
comparator $(i,j)$ does not do anything. So when constructing
$\sorts(C^d_n,\bar x_ 0)$
we can remove clauses in the $\update$ constraint that are
related to pairs $(i,j)$ touching leading zeroes or trailing ones.
In fact, these two optimizations  follow from the SAT encoding
described in Section~\ref{backgr_sat} by unit propagation; however,
because of the sheer size of the encoding (which for $n>16$ involves
millions of clauses), it is beneficial to perform these optimizations
at encoding time.

The third optimization is about adding additional, redundant, clauses
that help propagate information in the SAT solving phase. We
demonstrate this (and also the first two optimizations) by means of an
example.

\begin{example}\label{ex:improved}
  Consider a sorting network on $n=6$ channels, an input $\bar x^0 =
  (0, 1, 0, 1, 0, 1)$, and the output of the first layer, $\bar x^1 =
  (x^1_1, x^1_2, x^1_3, x^1_4, x^1_5, x^1_6)$.  
  Figure~\ref{fig:sort-6}(a) illustrates this setting, where a ``?'' on
  an input value indicates that we do not know whether a comparator will be
  placed somewhere on the corresponding channel.  By
  Lemma~\ref{lem:sorted}, $x^1_1=0$ and $x^1_6=1$, so $\bar x^1 = (0,
  x^1_2, x^1_3, x^1_4, x^1_5, 1)$, as indicated in the second layer of
  Figure~\ref{fig:sort-6}(a).
  Now consider the value of $x^1_4$. Clearly, the only first level
  comparator that will change the input value $x^0_4=1$ is $(4, 5)$.
  Therefore, adding any other comparator on channel $5$ determines that
  $x^1_4=1$ and one could specifically add propagation clauses of the
  form $g^1_{i,5}\rightarrow x^1_4$ for $1 \leq i < 4$.
  
  Figure~\ref{fig:sort-6}(b) illustrates the situation where
  comparator $(3,5)$ is placed in layer $1$. Channels $3$ and $5$ are
  now in use, hence the ``?'' is removed from the corresponding input
  values. The values of $x^1_3$ and $x^1_5$ are determined by the
  comparator. Moreover, as argued above, the value of $x^1_4$ is set to~$1$.
  
  Figure~\ref{fig:sort-6}(c) illustrates the situation if a second
  comparator, $(1,4)$, is added to layer $1$. The value $x^1_2=1$ is
  determined  by an argument similar to the one that determined
  $x^1_4=1$.  
\end{example}

\begin{figure}
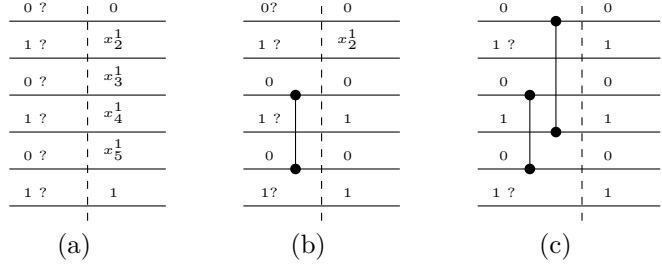

 \centering
\begin{tabular}{ccc}
 \begin{sortingnetwork}{6}{0.7}
   \nodelabel{\tiny 1   ?, \tiny 0 ?,\tiny 1 ?,\tiny 0 ?,\tiny 1 ?, \tiny 0 ?}
   \nextlayer
  \nodelabel{\tiny 1,\tiny $x^1_5$, \tiny $x^1_4$, \tiny $x^1_3$, \tiny $x^1_2$, \tiny 0}
 \end{sortingnetwork} &
  \begin{sortingnetwork}{6}{0.7}
   \nodelabel{\tiny 1?, \tiny 0,\tiny 1 ?,\tiny 0,\tiny 1 ?, \tiny 0?}
   \addcomparator24
   \nextlayer
  \nodelabel{\tiny 1,\tiny $0$, \tiny $1$, \tiny $0$, \tiny $x^1_2$, \tiny 0}
 \end{sortingnetwork} &
  \begin{sortingnetwork}{6}{0.7}
   \nodelabel{\tiny 1 ?, \tiny 0,\tiny 1 ,\tiny 0 ,\tiny 1 ?, \tiny 0}
   \addcomparator24
   \addlayer
   \addcomparator36
   \nextlayer
  \nodelabel{\tiny 1,\tiny $0$, \tiny $1$, \tiny $0$, \tiny $1$, \tiny 0}
 \end{sortingnetwork} \\
 (a) & (b) & (c) 
\end{tabular}

\caption{Propagations for the first layer of a sorting network on $6$
  channels determined from the input sequence $010101$ (see
  Example~\ref{ex:improved}).}
 \label{fig:sort-6}
\end{figure}

Example~\ref{ex:improved} demonstrates that positioning of single
comparators has additional implications that can be considered in a
SAT encoding. These implications derive from specific input sequences,
but the information required to use them is ``global''. 
Thus, for every layer $k$ and every pair of 
channels $(i,j)$ we introduce propositional variables $\OD^k_{i, j}$ and
$\OU^k_{i,j}$, which indicate whether there is a comparator
$g^k_{\ell,j}$ for some $i \leq \ell < j$ or $g^k_{i, \ell}$ for some
$i < \ell \leq j$, respectively.
\begin{align*}
   \OD^k_{i, j} & \leftrightarrow \bigvee_{i < \ell \leq j} g^k_{i,\ell} &
   \ND^k_{i, j} & \leftrightarrow \neg \OD^k_{i, j} \\
   \OU^k_{i, j} &\leftrightarrow \bigvee_{i \leq \ell < j} g^k_{\ell, j} & 
   \NU^k_{i, j} &\leftrightarrow \neg \OU^k_{i, j}
\end{align*}
To make use of these new propositional variables, given an input $\vec
x = (0, 0, \ldots, 0, x_t, x_{t+1}, \ldots, x_{t+r-1}, 1, 1, \ldots,
1) $, for all $t \leq i \leq t+r-1$ and at each layer $k$, we add the
following constraints to the definition of $\sorts$~(\ref{eq:sorts}).
\begin{align*}
 \bigwedge_{1 \leq k \leq d } v^{k-1}_i \wedge \ND^k_{i, t+r-1} &\rightarrow v^{k}_i \\
 \bigwedge_{1 \leq k \leq d } \neg v^{k-1}_i \wedge \NU^k_{t, i} &\rightarrow \neg v^{k}_i  
\end{align*}
This encoding allows for more propagations, thus conflicts can be
found earlier. 

\section{The End Game}
\label{sec:last}

Recent results on depth-optimal sorting networks stemmed from
the restriction of the search space of all comparator networks based
on complete sets of filters.
In this section, we focus on the dual problem: what do the \emph{last}
layers of a sorting network look like?
We obtain several interesting theoretical results that shed some
insight on the semantics of sorting networks, and also help
to solve open instances of the optimal depth-problem.

To the best of our knowledge, the only previous attempt to address this issue
was Parberry's heuristic for deciding whether a given prefix could be extended
to a sorting network~\cite{Parberry89,Parberry91}, which used necessary conditions on the
last two layers.  However, those conditions are not expressible without knowing
what all but the last two layers are.
In contrast, we look for generic conditions about all sorting networks.

\subsection{The Last Layers of a Sorting Network}
\label{sec:props}

We begin by recalling the notion of redundant comparator, introduced
in Exercise~5.3.4.51 of~\cite{Knuth73} with credit to R.L.~Graham.
\begin{definition}
  Let $C;(i,j);C'$ be a comparator network.
  The comparator $(i,j)$ is \emph{redundant} if $\vec x_i\leq\vec x_j$
  for all $\vec x\in\outputs(C)$.
\end{definition}
A sorting network without redundant comparators is called \emph{non-redundant}.

\begin{lemma}
  \label{lem:redundant}
  Let $D$ be a comparator network on $n$ channels.
  Construct $D'$ by removing every redundant comparator from $D$.
  Then $D'$ is a sorting network iff $D$ is a sorting network.
\end{lemma}
\begin{proof}
  From the definition of redundant comparator, it follows that
  $D(\vec x)=D'(\vec x)$ for every $\vec x\in\{0,1\}^n$.
\end{proof}

Lemma~\ref{lem:redundant} was already explored in the proof of optimality of the
$25$-comparator sorting network on $9$~channels~\cite{ourICTAIpaper}.
In general, it is not easily applicable to proofs of optimal depth based on
SAT-encodings, because redundancy is a semantic property that is not easy to
encode syntactically.
In order to use it in this context, we need syntactic criteria for
redundancy.

\begin{lemma}
  \label{lem:lastlayer}
  Let $C$ be a non-redundant sorting network on $n$ channels.
  Then all comparators in the last layer of $C$ are of the form
  $(i,i+1)$.
\end{lemma}

\begin{wrapfigure}[7]{r}{0.35\textwidth}
  \begin{tabular}{cc}
    \begin{sortingnetwork}{3}{0.7}
      \nodelabel{0,0,1}
      \addcomparator13
      \addlayer
      \nodelabel{1,0,0}
    \end{sortingnetwork}
    &
    \begin{sortingnetwork}{3}{0.7}
      \nodelabel{0,1,1}
      \addcomparator13
      \addlayer
      \nodelabel{1,1,0}
    \end{sortingnetwork}
    \\
    $(a)$ & $(b)$
  \end{tabular}
\end{wrapfigure}\ \vspace*{-1.8em}

\begin{proof}
  Let $C$ be a non-redundant sorting network with a comparator
  $c=(i,i+2)$ in the last layer.
  Since $c$ is not redundant, there is an input $\vec x$ such that
  channels $i$ to $i+2$ before applying $c$ look like~$(a)$ or~$(b)$
  on the right.

  Suppose $\vec x$ is an input yielding case~$(a)$, and let $\vec y$
  be any input obtained by replacing one $0$ in $\vec x$ by a $1$.
  Since $C$ is a sorting network, $C(\vec y)$ is sorted, but since
  $\vec x<\vec y$ the value in channel $i$ before applying $c$ must be
  a $1$ (Lemma~\ref{lem:monot}), hence $\vec y$ yields
  situation~$(b)$.
  Dually, given $\vec y$ yielding $(b)$, we know that any $\vec z$
  obtained by replacing one $1$ in $\vec y$ by a $0$ will yield
  $(a)$.

  Thus all inputs with the same number of zeroes as $\vec x$ or
  $\vec y$ must yield either~$(a)$ or~$(b)$, in particular sorted
  inputs, contradicting Lemma~\ref{lem:sorted}.
  The same reasoning generalizes to show that $c$ cannot have the form $(i,i+k)$
  with $k\geq2$, thus it has to be of the form $(i,i+1)$.
\end{proof}

\begin{corollary}
  \label{cor:lastlayer}
  Suppose that $C$ is a non-redundant sorting network that contains a
  comparator $(i,j)$ at layer~$d$, with $j>i+1$.
  Then at least one of channels $i$ and $j$ is used in a layer $d'$
  with $d'>d$.
\end{corollary}
\begin{proof}
  If neither $i$ nor $j$ are used after layer~$d$, then comparator
  $(i,j)$ can be moved to the last layer without changing the function
  computed by~$C$.  By the previous lemma, $C$ cannot be a
  non-redundant sorting network.
\end{proof}

Lemma~\ref{lem:lastlayer} restricts the number of possible comparators
in the last layer in a sorting network on $n$ channels to $n-1$,
instead of $n(n-1)/2$ in the general case.

\begin{theorem}
  \label{thm:lastlayer}
  The number of possible last layers in an $n$-channel sorting network
  with no redundancy is $L_n=F_{n+1}-1$, where $F_n$ denotes the
  Fibonacci sequence.
\end{theorem}
\begin{proof}
  Denote by $L^+_n$ the number of possible last layers on $n$
  channels, where the last layer is allowed to be empty (so
  $L_n=L^+_n-1$).  There is exactly one possible last layer on $1$
  channel, and there are two possible last layers on $2$ channels (no
  comparators or one comparator), so $L^+_1=F_2$ and $L^+_2=F_3$.

  Given a layer on $n$ channels, there are two possibilities.  Either
  the first channel is unused, and there are $L^+_{n-1}$ possibilities
  for the remaining $n-1$ channels; or it is connected to the second
  channel, and there are $L^+_{n-2}$ possibilities for the remaining
  $n-2$ channels.
  So $L^+_n=L^+_ {n-1}+L^+_{n-2}$, whence $L^+_n=F_{n+1}$.
\end{proof}

Even though $L_n$ grows quickly, it grows much slower than the number
$|G_n|$ of possible layers in general (Table~\ref{tab:Rn}).
In particular, $L_{17}=2583$, whereas $G_{17}=211{,}799{,}312$.

To move backwards from the last layer, we introduce an auxiliary notion.
\begin{definition}
  Let $C$ be a depth $d$ sorting network without redundant
  comparators, and let $k<d$.
  A \emph{$k$-block} of $C$ is a maximal set of channels $B$ such that:
  if $i,j\in B$, then there is a sequence of channels
  $i=x_0,\ldots,x_\ell=j$
  where $(x_i,x_{i+1})$ or $(x_{i+1},x_i)$ is a comparator in a layer
  $k'>k$ of $C$.
\end{definition}
For each $k$ the set of $k$-blocks of $C$ is a partition of the set of
channels of $C$.
Given a $k$-block $B$, we will abuse terminology and refer to ``the
comparators in $B$'' to denote the comparators connecting channels in~$B$
after layer $k$.

Given a comparator network of depth $d$, we will call its
$(d-1)$-blocks simply \emph{blocks} -- so Lemma~\ref{lem:lastlayer}
states that a block in a sorting network $C$ contains either (a)~a
channel unused at the last layer of $C$ or (b)~two adjacent channels
connected by a comparator at the last layer of $C$.

\begin{example}
  Recall the sorting network shown in Figure~\ref{fig:sort-5}, and reproduced in Figure~\ref{fig:blocks}.
  Its $4$-blocks, or simply blocks, are $\{1\}$, $\{2\}$, $\{3,4\}$
  and $\{5\}$;
  its $3$-blocks are $\{1\}$, $\{2,3,4,5\}$;
  and for $k<3$ there is only the trivial $k$-block
  $\{1,2,3,4,5\}$ .
  The $3$-block $\{2,3,4,5\}$ contains the comparators $(2,3)$, $(4,5)$ and
  $(3,4)$.
\end{example}

\begin{figure}[htb]
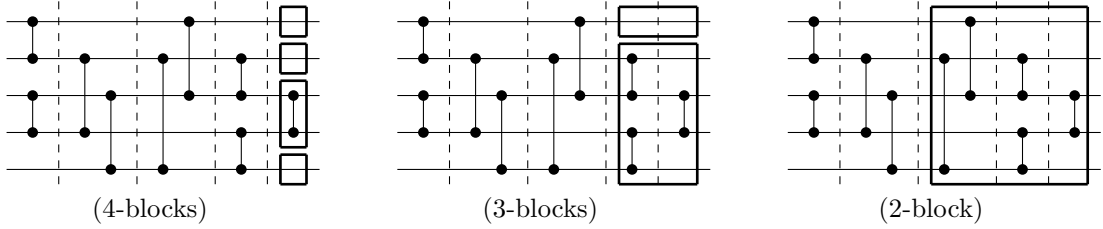

  \centering
  \begin{tabular}{ccc}
    \begin{sortingnetwork}{5}{0.7}
      \addcomparator45\addcomparator23\nextlayer
      \addcomparator24\addlayer\addcomparator13\nextlayer
      \addcomparator14\addlayer\addcomparator35\nextlayer
      \addcomparator12\addcomparator34\nextlayer
      \addcomparator23
      \draw[line width=1pt](\sncolwidth*10.5,0.6)--(\sncolwidth*10.5,1.4);
      \draw[line width=1pt](\sncolwidth*10.5,1.4)--(\sncolwidth*9.5,1.4);
      \draw[line width=1pt](\sncolwidth*9.5,1.4)--(\sncolwidth*9.5,0.6);
      \draw[line width=1pt](\sncolwidth*9.5,0.6)--(\sncolwidth*10.5,0.6);
      \draw[line width=1pt](\sncolwidth*10.5,1.6)--(\sncolwidth*10.5,3.4);
      \draw[line width=1pt](\sncolwidth*10.5,3.4)--(\sncolwidth*9.5,3.4);
      \draw[line width=1pt](\sncolwidth*9.5,3.4)--(\sncolwidth*9.5,1.6);
      \draw[line width=1pt](\sncolwidth*9.5,1.6)--(\sncolwidth*10.5,1.6);
      \draw[line width=1pt](\sncolwidth*10.5,3.6)--(\sncolwidth*10.5,4.4);
      \draw[line width=1pt](\sncolwidth*10.5,4.4)--(\sncolwidth*9.5,4.4);
      \draw[line width=1pt](\sncolwidth*9.5,4.4)--(\sncolwidth*9.5,3.6);
      \draw[line width=1pt](\sncolwidth*9.5,3.6)--(\sncolwidth*10.5,3.6);
      \draw[line width=1pt](\sncolwidth*10.5,4.6)--(\sncolwidth*10.5,5.4);
      \draw[line width=1pt](\sncolwidth*10.5,5.4)--(\sncolwidth*9.5,5.4);
      \draw[line width=1pt](\sncolwidth*9.5,5.4)--(\sncolwidth*9.5,4.6);
      \draw[line width=1pt](\sncolwidth*9.5,4.6)--(\sncolwidth*10.5,4.6);
    \end{sortingnetwork}
    &
    \begin{sortingnetwork}{5}{0.7}
      \addcomparator45\addcomparator23\nextlayer
      \addcomparator24\addlayer\addcomparator13\nextlayer
      \addcomparator14\addlayer\addcomparator35\nextlayer
      \addcomparator12\addcomparator34\nextlayer
      \addcomparator23
      \draw[line width=1pt](\sncolwidth*10.5,0.6)--(\sncolwidth*10.5,4.4);
      \draw[line width=1pt](\sncolwidth*10.5,4.4)--(\sncolwidth*7.5,4.4);
      \draw[line width=1pt](\sncolwidth*7.5,4.4)--(\sncolwidth*7.5,0.6);
      \draw[line width=1pt](\sncolwidth*7.5,0.6)--(\sncolwidth*10.5,0.6);
      \draw[line width=1pt](\sncolwidth*10.5,4.6)--(\sncolwidth*10.5,5.4);
      \draw[line width=1pt](\sncolwidth*10.5,5.4)--(\sncolwidth*7.5,5.4);
      \draw[line width=1pt](\sncolwidth*7.5,5.4)--(\sncolwidth*7.5,4.6);
      \draw[line width=1pt](\sncolwidth*7.5,4.6)--(\sncolwidth*10.5,4.6);
    \end{sortingnetwork}
    &
    \begin{sortingnetwork}{5}{0.7}
      \addcomparator45\addcomparator23\nextlayer
      \addcomparator24\addlayer\addcomparator13\nextlayer
      \addcomparator14\addlayer\addcomparator35\nextlayer
      \addcomparator12\addcomparator34\nextlayer
      \addcomparator23
      \draw[line width=1pt](\sncolwidth*10.5,0.6)--(\sncolwidth*10.5,5.4);
      \draw[line width=1pt](\sncolwidth*10.5,5.4)--(\sncolwidth*4.5,5.4);
      \draw[line width=1pt](\sncolwidth*4.5,5.4)--(\sncolwidth*4.5,0.6);
      \draw[line width=1pt](\sncolwidth*4.5,0.6)--(\sncolwidth*10.5,0.6);
    \end{sortingnetwork}
\\
 ($4$-blocks) & ($3$-blocks) & ($2$-block)
 \end{tabular}
 \caption{The $4$-, $3$- and $2$-blocks of the sorting network
   $\{(1,2),(3,4)\},\{(2,4),(3,5)\},\{(1,3),(2,5)\},\{(2,3),(4,5)\},\{(3,4)\}$.}
 \label{fig:blocks}
\end{figure}

The notion of block helps understand the semantics of a sorting
network: after layer $k$, all $k$-blocks are sorted except for at most
one.

\begin{lemma}
  \label{lem:k-block}
  Let $C$ be a sorting network on $n$ channels with depth $d$, and $k<d$.
  For each input $\vec x\in\{0,1\}^n$, there is at most one $k$-block
  that receives both $0$s and $1$s as input.
\end{lemma}
\begin{proof}
  From the definition of $k$-block, it follows that values cannot move
  between two different $k$-blocks.
  Therefore, for every input, if there is a $k$-block that
  receives both $0$s and $1$s as inputs, then all $k$-blocks above it must
  receive only $0$s and those below it must receive only $1$s in order for
  the output to be sorted.
\end{proof}
For example, in Figure~\ref{fig:blocks} (left), if the $4$-block $\{3,4\}$
receives one $0$ and one $1$, then the $4$-blocks $\{1\}$ and $\{2\}$ must
receive a $0$ and the $4$-block $\{5\}$ must receive a $1$.

\begin{lemma}
  \label{lem:layerbeforelast}
  Let $C$ be a non-redundant sorting network on $n$ channels with
  depth $d$.
  Then all comparators in layer $d-1$ connect adjacent blocks of $C$.
\end{lemma}
\begin{proof}
  The proof is similar to that of Lemma~\ref{lem:lastlayer}, but now
  considering blocks instead of channels.
  Let $c$ be a comparator in layer $d-1$ of $C$ that does not connect
  adjacent blocks of $C$.
  Observe that $c$ cannot connect two channels in the same block, as
  then there would be two copies of $c$ in the last two layers, and
  the second one would be redundant.

  Since $c$ is not redundant, there must be some input $\bar x$ that
  provides $c$ with input $1$ on its top channel and $0$ on its bottom
  channel.
  The situation is depicted below, where $X$ and $Z$ are blocks, and
  $Y$ is the set of channels in between.
  According to Lemma~\ref{lem:k-block}, there are five possible cases
  for $X$, $Y$ and $Z$,
  depending on the number of $0$s in $\bar x$.

  \begin{center}
    \begin{sortingnetwork}{3}{0.7}
      \nodelabel{0,,1}
      \addcomparator13
      \addlayer
      \nodelabel{Z,Y,X}
    \end{sortingnetwork}
    \qquad
    \begin{tabular}[b]{c|c|c|c|c|c}
      X & all $0$s & all $0$s & all $0$s & all $0$s & mixed \\
      Y & all $0$s & all $0$s & mixed & all $1$s & all $1$s \\
      Z & mixed & all $1$s & all $1$s & all $1$s & all $1$s \\ \midrule
     &$(a)$&$(b)$&$(c)$&$(d)$&$(e)$
    \end{tabular}
  \end{center}

  Suppose that input $\vec x$ leads to case~$(a)$.
  By changing the appropriate number of $0$s in $\vec x$ to $1$s, we
  can find an input $\vec y$ that leads to case $(b)$, since again by
  monotonicity of $C$ $\vec y$ cannot place a $0$ to the top input of
  $c$.
  Likewise, we can reduce $(e)$ to $(d)$.
  But now we can move between $(b)$, $(c)$ and $(d)$ by changing one
  bit of the word at a time.
  By Lemma~\ref{lem:monot}, this must keep either the top $1$ input of
  $c$ or the lower $0$, while the other value is kept by the fact that
  $C$ is a sorting network.
  As in Lemma~\ref{lem:lastlayer}, this proves that this configuration
  occurs for all words with the same number of $0$s, which is absurd
  since it cannot happen for sorted inputs.
\end{proof}

Combining this result with Lemma~\ref{lem:lastlayer} we obtain the
explicit configurations that can occur in a sorting network.
\begin{corollary}
  \label{cor:layerbeforelast}
  Let $C$ be a non-redundant sorting network on $n$ channels with
  depth $d$.
  Then every comparator $(i,j)$ in layer $d-1$ of $C$ satisfies
  $j-i\leq 3$.
  Furthermore, if $j=i+2$, then either $(i,i+1)$ or $(i+1,i+2)$ occurs
  in the last layer;
  and if $j=i+3$, then both $(i,i+1)$ and $(i+2,i+3)$ occur in the
  last layer.
\end{corollary}

The sorting networks in Figure~\ref{fig:tight} show that the bound $j-i\leq 3$ is tight.
\begin{figure}
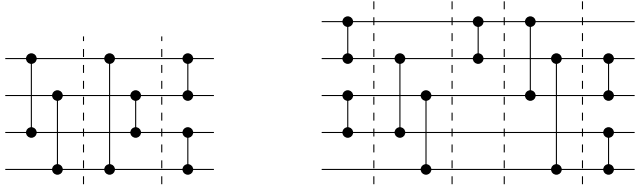

  \centering
  \begin{sortingnetwork}4{0.7}
    \addcomparator24
    \addlayer
    \addcomparator13
    \nextlayer
    \addcomparator14
    \addlayer
    \addcomparator23
    \nextlayer
    \addcomparator12
    \addcomparator34
  \end{sortingnetwork}
  \qquad
  \begin{sortingnetwork}5{0.7}
    \addcomparator45
    \addcomparator23
    \nextlayer
    \addcomparator24
    \addlayer
    \addcomparator13
    \nextlayer
    \addcomparator45
    \nextlayer
    \addcomparator35
    \addlayer
    \addcomparator14
    \nextlayer
    \addcomparator12
    \addcomparator34
  \end{sortingnetwork}
  \caption{Sorting networks containing a comparator $(i,i+3)$ in their penultimate layer.}
  \label{fig:tight}
\end{figure}

\begin{theorem}
  \label{thm:layerbeforelast}
  If $C$ is a non-redundant sorting network on $n$ channels and
  $(i,j)$ is a comparator at layer $k$ of $C$, then $i$ and $j$ are
  either in the same $k$-block or in adjacent $k$-blocks of~$C$.
\end{theorem}
\begin{proof}
  As for Lemma~\ref{lem:layerbeforelast}, considering $k$-blocks instead of blocks.
\end{proof}

As a consequence of Theorem~\ref{thm:layerbeforelast}, all $k$-blocks in a
non-redundant sorting network are sets of consecutive channels.

\begin{corollary}
  \label{cor:blocksize}
  Let $C$ be a non-redundant sorting network on $n$ channels and $B$
  be a $k$-block of $C$ containing $m$ comparators.
  Then $B$ consists of at most $m+1$ channels.
\end{corollary}
\begin{proof}
  By induction on $m$.
  The case $m=1$ is simply Corollary~\ref{cor:lastlayer}.
  If $m>1$, let $c$ be a comparator at layer $k+1$ of $C$ and consider
  the sorting network $C'$ obtained by placing $c$ into an individual
  layer.
  By Theorem~\ref{thm:layerbeforelast}, either (i)~$c$ connects
  channels within a $(k+1)$-block $B'$ of $C'$ or (ii)~$c$ connects
  channels in adjacent $(k+1)$-blocks $B_1$ and $B_2$ of $C'$.
  In case~(i), block $B'$ contains $m-1$ comparators, hence by
  induction hypothesis it consists of at most $m-2$ channels; since
  $c$ does not add any new channels, the thesis also holds for $B$.
  In case~(ii), blocks $B_1$ and $B_2$ contain $m_1$ and $m_2$
  comparators each, with $m_1+m_2=m-1$, and by induction hypothesis
  they consist of at most $m_1+1$ and $m_2+1$ channels, respectively.
  Since $B$ consists of the union of these sets, it has at most
  $(m_1+1)+(m_2+1)=m+1$ elements.
\end{proof}

\subsection{Co-Saturation}
\label{sec:co-sat}

Using the results from Section~\ref{sec:props}, we can reduce the search space
of possible sorting networks of a given depth simply by restricting to
comparator networks satisfying the necessary conditions presented, namely
Lemma~\ref{lem:lastlayer} and Corollary~\ref{cor:layerbeforelast}.
However, the successful strategies
in~\cite{BZ14,ourICTAIpaper,Parberry89,Parberry91} all focus on also imposing
\emph{sufficient} conditions on those networks.
This is expressed by results of the form ``if there is a sorting
network, then there is one satisfying this property''.

We now follow this approach pursuing the idea of saturation
from~\cite{BCCSZ14}: how many (redundant) comparators can we safely
add to the last layers of a sorting network?
We will show how to do this in a structured way that actually reduces
the number of possibilities for the last two layers.
Again, we capitalize on the observation that redundant comparators do
not change the function represented by a comparator network and can,
thus, be removed or added at will.

\begin{lemma}
  \label{lem:llnf}
  Let $C$ be a sorting network on $n$ channels.
  There is a sorting network $N$ of the same depth as $C$ whose last
  layer: (i)~only contains comparators between adjacent channels;
  and (ii)~does not contain two adjacent unused channels.
\end{lemma}
\begin{proof}
  Let $C$ be a sorting network on $n$ channels.
  By Lemma~\ref{lem:redundant}, we can eliminate all redundant
  comparators from $C$ to obtain a sorting network $S$.
  By Lemma~\ref{lem:lastlayer}, all comparators in the last layer of
  $S$ are of the form $(i,i+1)$.
  Let $j$ be such that $j$ and $j+1$ are unused in the last layer of
  $S$; since $S$ is a sorting network, this means that comparator
  $(j,j+1)$ is redundant, and we can add it to the last layer of $S$.
  Repeating this process for $j=1,\ldots,n$ we obtain a sorting
  network $N$ that satisfies both desired properties.
\end{proof}
We say that a sorting network satisfying the conditions of
Lemma~\ref{lem:llnf} is in \emph{last layer normal form} (llnf).

\begin{theorem}
  \label{thm:llnf}
  The number of possible last layers in llnf on $n$ channels is
  $K_n=P_{n+5}$, where $P_n$ denotes the Padovan sequence, defined as
  $P_0=1$, $P_1=P_2=0$ and $P_{n+3}=P_n+P_{n+1}$.
\end{theorem}
\begin{proof}
  Let $K^+_n$ be the number of layers in llnf that begin with the
  comparator $(1,2)$, and $K^-_n$ the number of those where channel
  $1$ is free. Then $K_n=K^+_n+K^-_n$.
  Let $n>3$.
  If a layer in llnf begins with a comparator, then there are
  $K_{n-2}$ possibilities for the remaining channels; if it begins
  with a free channel, then there are $K^+_{n-1}$ possibilities for
  the remaining channels.
  Therefore $K_n = K^+_n+K^-_n = K_{n-2}+K^+_{n-1} = K_{n-2}+K_{n-3}$.
  There exist one last layer on $1$ channel (with no comparator),
  one on $2$ channels (with one comparator between them) and two on $3$ channels
  (one comparator between either the top two or the bottom two channels),
  so $K_1=P_6$, $K_2=P_7$ and $K_3=P_8$.
  From the recurrence it follows that $K_n=P_{n+5}$.
\end{proof}

Sequence $K_n$ grows much slower than the total number $L_n$ of
non-redundant last layers identified in Theorem~\ref{thm:lastlayer}.
For example, $K_{17}=86$ whereas $L_{17}=2583$.

Given that the last layer is required to be in llnf, we can also study
the previous layer.
By Lemma~\ref{lem:layerbeforelast}, we know that every block can only
be connected to the adjacent ones; again we can \emph{add} redundant
comparators to reduce the number of possibilities for the last two
layers.

\begin{lemma}
  \label{lem:twolast}
  Let $C$ be a sorting network of depth $d$ in llnf.  Let $i<j$ be two
  channels that are unused in layer $d-1$ and that belong to different
  blocks.  Then adding the comparator $(i,j)$ to layer $d-1$ of $C$
  still yields a sorting network.
\end{lemma}
\begin{proof}
  Suppose there is an input $\bar x$ such that channel $i$ carries a
  $1$ at layer $d-1$, and channel $j$ carries a $0$ at that same
  layer.
  Since neither channel is used, their corresponding blocks will
  receive these values.
  But then $C(\bar x)$ has a $1$ in a channel in the block containing
  $i$ and a $0$ in the block containing $j$, and since $i<j$ this
  sequence is not sorted by $C$, contradicting the assumption that $C$
  is a sorting network.
  Therefore the comparator $(i,j)$ at layer $d-1$ of $C$ is redundant,
  and can be added to this network.
\end{proof}

Incidentally, this lemma provides a partial answer to a problem
also posed in Exercise~5.3.4.21~\cite{Knuth73}: when can we add
comparators to a sorting network while keeping it sorting?
It was already known that this can always be done in the first and
last layers, while a result by N.G.~de Bruijn~\cite{deBruijn74} states
that this can always be done in networks that only contain comparators
of the form $(i,i+1)$.
Lemma~\ref{lem:twolast} gives a general sufficient condition, and it
implies in particular that comparators can also always be added
between any two free channels in the penultimate layer of any sorting
network.

\begin{lemma}
  \label{lem:onetwo}
  Let $C$ be a sorting network of depth $d$ in llnf.
  Suppose that there is a comparator $(i,i+1)$ in the last layer of
  $C$, that channel $i+2$ is used in layer $d-1$ but not in layer $d$,
  and that channels $i$ and $i+1$ are both unused in layer $d-1$ of
  $C$ (see Figure~\ref{fig:onetwo} (a)).
  Then there is a sorting network $C'$ of depth $d$ in llnf such that
  channels $i+1$ and $i+2$ are both used in layers $d-1$ and $d$.
\end{lemma}
\begin{proof}
  Since channels $i$ and $i+1$ are unused in layer $d-1$, comparator
  $(i,i+1)$ can be moved to that layer without changing the behavior
  of $C$; then the redundant comparator $(i+1,i+2)$ can be added to
  layer $d$, yielding the sorting network $C'$
  (Figure~\ref{fig:onetwo} (a)).
  If $i>1$ and channel $i-1$ is not used in the last layer of $C$,
  then $C'$ must also contain a comparator $(i-1,i)$ in its last layer
  (Figure~\ref{fig:onetwo} (b)).
\end{proof}

\begin{figure}
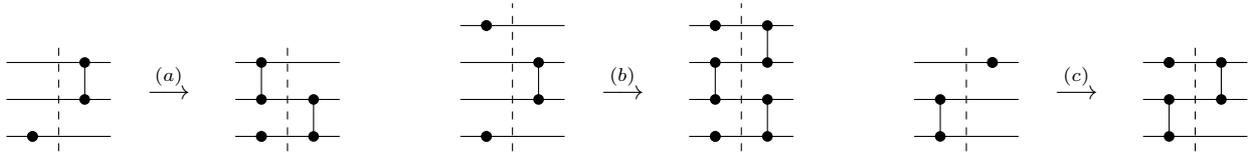

  \centering
  \begin{sortingnetwork}{3}{0.7}
    \addcomparator11
    \nextlayer
    \addcomparator23
  \end{sortingnetwork}
  \raisebox{2em}{\ \ $\stackrel{(a)}\longrightarrow$}
  \begin{sortingnetwork}{3}{0.7}
    \addcomparator11
    \addcomparator23
    \nextlayer
    \addcomparator12
  \end{sortingnetwork}
  \hfill
  \begin{sortingnetwork}{4}{0.7}
    \addcomparator11
    \addcomparator44
    \nextlayer
    \addcomparator23
  \end{sortingnetwork}
  \raisebox{2em}{\ \ $\stackrel{(b)}\longrightarrow$}
  \begin{sortingnetwork}{4}{0.7}
    \addcomparator11
    \addcomparator23
    \addcomparator44
    \nextlayer
    \addcomparator12
    \addcomparator34
  \end{sortingnetwork}
  \hfill
  \begin{sortingnetwork}{3}{0.7}
    \addcomparator12
    \nextlayer
    \addcomparator33
  \end{sortingnetwork}
  \raisebox{2em}{\ \ $\stackrel{(c)}\longrightarrow$}
  \begin{sortingnetwork}{3}{0.7}
    \addcomparator33
    \addcomparator12
    \nextlayer
    \addcomparator23
  \end{sortingnetwork}

  \caption{Transformations in the proof of Lemma~\ref{lem:onetwo}.}
  \label{fig:onetwo}
\end{figure}
A construction similar to that in Lemma~\ref{lem:onetwo} can also be
applied if channel $i-1$ (instead of $i+2$) is used at layer $d-1$ and
unused in layer $d$ (Figure~\ref{fig:onetwo} (c)).

\begin{definition}
  \label{defn:cosat}
  A sorting network of depth $d$ is \emph{co-saturated} if:
  (i)~its last layer is in llnf,
  (ii)~no two consecutive blocks at layer $d-1$ have unused channels,
  and (iii)~if $(i,i+1)$ is a comparator in layer $d$ and channels $i$
  and $i+1$ are unused in layer $d-1$, then channels $i-1$ and $i+2$
  (if they exist) are used in layer~$d$.
\end{definition}

\begin{theorem}
  \label{thm:cosat}
  If $C$ is a sorting network on $n$ channels with depth $d$, then
  there is a co-saturated sorting network $N$ on $n$ channels with
  depth $d$.
\end{theorem}
\begin{proof}
  Assume $C$ is given.
  Apply Theorem~\ref{thm:llnf} to find a sorting network $S$ in llnf,
  containing no redundant comparators except possibly in the last
  layer.

  Let $B_1,\ldots,B_k$ be the $(d-1)$-blocks in $S$.
  For $i=1,\ldots,k-1$, if blocks $B_i$ and $B_{i+1}$ have a free
  channel, add a comparator between them.
  (It may be possible to add \emph{two} comparators between these
  blocks, namely if they both have two channels and none is used in
  layer~$d-1$.)
  Let $N$ be the resulting network.
  By Lemma~\ref{lem:twolast}, all the comparators added from $S$ to
  $N$ are redundant, so $N$ is a sorting network; by construction, $N$
  satisfies~(ii).

  If $N$ does not satisfy~(iii), then applying Lemma~\ref{lem:onetwo}
  transforms it into another sorting network $N'$ that does.
\end{proof}

In order to get a quantitative measure of the reduction of the search
space obtained by adding these constraints, we wrote a simple computer
program to generate all co-saturated two-layer suffixes on $n$
channels.
The values are given in Table~\ref{tab:co-sat}.

\begin{table}
  \centering
  \[\begin{array}{c|cccccccccccccc}
    \toprule
    n & 3 & 4 & 5 & 6 & 7 & 8 & 9 & 10 & 11 & 12 & 13 & 14 & 15 & 16 \\ \midrule
    \# & 4 & 4 & 12 & 26 & 44 & 86 & 180 & 376 & 700 & 1{,}440 &
    2{,}892 & 5{,}676 & 11{,}488 & 22{,}848 \\ \bottomrule
  \end{array}\]
  \[\begin{array}{c|cccc}
  \toprule
    n & 17 & 18 & 19 & 20 \\ \midrule
    \# & 45{,}664 & 90{,}976 & 182{,}112 & 363{,}896 \\ \bottomrule
  \end{array}\]

  \caption{Number of distinct co-saturated two-layer suffixes on $n$ channels, for $n\leq 20$.}
  \label{tab:co-sat}
\end{table}

\subsection{Co-Subsumption}
\label{sec:cosub}

We now explore a stronger restriction of the search space, based on
a different dualization of the ideas in~\cite{BCCSZ14}.
Given a comparator network $C$, define
$\inputs(C)=\{\vec x\in\{0,1\}^n\mid C(\vec x)\mbox{ is sorted}\}$.
We say that $C$ \emph{co-subsumes} $C'$, $C\cosub C'$, if there exists
a permutation $\pi$ such that $\pi(\inputs(C))\subseteq\inputs(C')$.
When we want to make the permutation $\pi$ explicit, we will write
$C\cosub_\pi C'$.
A \emph{generalized comparator} is a comparator $(i,j)$ where $i>j$,
and a \emph{generalized comparator network} is a comparator network
that may contain generalized comparators.
A \emph{generalized sorting network} is a generalized comparator
network that sorts all binary inputs.
When relevant, we will refer to a \emph{standard} comparator network
to emphasize that we are referring to a comparator network with no
generalized comparators.

\begin{lemma}
  \label{lem:cosub}
  If $C\cosub C'$ and there is a sorting network of the form $N;C$,
  then there is a generalized sorting network of the form $N';C'$ of
  the same size and depth.
\end{lemma}
\begin{proof}
  Construct $N'$ from $N$ by renaming the channels according to $\pi$.
  Given an input $\vec x$, $N(\vec x)\in\inputs(C)$, and therefore
  $\pi(N)(\pi(\vec x))\in\inputs(C')$.
\end{proof}

In particular, the following corollary states that we can fix the last
layer, as we could fix the first one.
\begin{corollary}
  \label{cor:cosub}
  There is an optimal-depth generalized sorting network whose last
  layer is $L^1_P$ (Equation~\eqref{eq:filterParberry}).
\end{corollary}
\begin{proof}
  Let $C$ be the last layer of an optimal-depth (generalized) sorting
  network.
  Then $C\cosub L^1_P$: if $C=\{(i_1,j_1),\ldots,(i_\ell,j_\ell)\}$, take
  $\pi$ to be the permutation mapping $i_k$ to $2k-1$ and $j_k$ to
  $2k$.
\end{proof}

Note that this result cannot be combined with e.g.~Parberry's result
that there always exists a sorting network whose first layer is $L^1_P$:
by fixing the first layer we dictate what the last layer must be, and
conversely.

In general, the networks constructed in these proofs can have reversed
comparators, and the usual transformation described in Exercise
5.3.4.16 of~\cite{Knuth73} for removing these cannot be applied, as it
changes the last layers.
However, we can also define a \emph{dual} standardization procedure
that can be applied in these cases.

\begin{definition}
  Let $C$ be a generalized sorting network, and define a standard
  comparator network as follows: let $k$ be the last layer of $C$
  where a reversed comparator $(i,j)$ occurs, and interchange $i$ and
  $j$ everywhere in layers $1$ to $k$.
  Iterate this procedure until there are no reverse comparators left;
  we denote the resulting network by $C^{st}$.
\end{definition}

\begin{lemma}
  If $C$ is a generalized sorting network, then $C^{st}$ is a
  sorting network.
\end{lemma}
\begin{proof}
  At each step, we are applying the construction in
  Lemma~\ref{lem:cosub} to the suffix of $C$ starting with comparator
  $(i,j)$.
\end{proof}

\begin{corollary}
  The sorting networks constructed in Lemma~\ref{lem:cosub} and
  Corollary~\ref{cor:cosub} can be taken to be standard.
\end{corollary}

In particular, there is always a sorting network on $n$ channels whose
last layer is $L^1_P$.

\subsection{Practical Impact}
\label{sec:cosat}

We now show how to improve the SAT encoding presented in
Section~\ref{backgr_sat} using the structure of the last layers of
a sorting network.  We first consider the impact of the necessary
conditions stated in Section~\ref{sec:props}.

Consider Lemma~\ref{lem:lastlayer}, which states that non-redundant
comparators in the last layer have to be of the form $(i,i+1)$.  The
following propositional constraint forbids comparators in the last
layer that connect non-adjacent channels. The constraint is expressed
in terms of a quadratic number of unit clauses:
\[\varphi_1 = \sset{\neg g^d_{i,j}}{1\leq i,i+1<j\leq n}\]

Corollary~\ref{cor:lastlayer} generalizes Lemma~\ref{lem:lastlayer},
stating that whenever a comparator at any layer connects two
non-adjacent channels, necessarily one of these channels is used at a
later layer. This is captured by the following propositional
constraint, which adds a single clause for each of the $(n-1)(n-2)/2$
non-adjacent comparators at any given depth $\ell$:
\[\varphi_1(\ell) = 
   \sset{g^\ell_{i,j}\to\bigvee_{\ell <k\leq d}\used^k_i\vee\used^k_j}
     {1\leq i,i+1<j\leq n}\]
Note that indeed $\varphi_1(d) = \varphi_1$, as there is no depth $k$
with $\ell < k \leq d$.

Consider now the penultimate layer, $d-1$.  According to
Corollary~\ref{cor:layerbeforelast}, no comparator at this layer can
connect two channels more than $3$ channels apart.  Similar to
Lemma~\ref{lem:lastlayer}, we encode this restriction by adding unit
clauses for each of the $(n-3)(n-4)/2$ comparators more than $3$
channels apart:
\[ \varphi_2 = \sset{\neg g^{d-1}_{i,j}}{1\leq i,i+3<j\leq n} \]

Corollary~\ref{cor:layerbeforelast} also states that the existence of
a comparator $(i,i+3)$ on the penultimate layer implies the existence
of the two comparators $(i,i+1)$ and $(i+2,i+3)$ on the last layer.
This is encoded introducing $2(n-3)$ clauses in the following
constraint :
\[ \varphi_3 = 
    \sset{g^{d-1}_{i,i+3} \to g^d_{i,i+1}\right) 
          \wedge \left(g^{d-1}_{i,i+3} \to g^d_{i+2,i+3}}{1\leq i\leq
            n-3}
\]

Corollary~\ref{cor:layerbeforelast} also states that the
existence of a comparator $(i,i+2)$ on the penultimate layer implies
the existence of either of the comparators $(i,i+1)$ or $(i+1,i+2)$ on
the last layer.
This can be encoded using $n-2$ clauses:
\[ \varphi_4 = \sset{g^{d-1}_{i,i+2}\to g^d_{i,i+1}\vee g^d_{i+1,i+2}}{1\leq i\leq n-2} \]

We now consider the sufficient conditions from
Section~\ref{sec:co-sat}.  According to Lemma~\ref{lem:llnf}~(ii), we
can break symmetries by requiring that there be no adjacent unused
channels in the last layer, i.e., that the network be in llnf.
\[ \psi_1 = \sset{\used^d_i\vee\used^d_{i+1}}{1\leq i<n} \]
Essentially, this forces the SAT solver to add a (redundant)
comparator between any two adjacent unused channels on the last
layer.

The next symmetry break is based on a consideration of two adjacent
blocks.
There are four possible cases: two adjacent comparators, a comparator
followed by an unused channel, an unused channel followed by a
comparator, and two unused channels.
The latter is forbidden by the symmetry break $\psi_1$ (and thus not
regarded further).

The case of two adjacent comparators is handled by formula $\psi_2^a$:
\[ \psi_2^a = 
     \sset{g^d_{i,i+1}\wedge g^d_{i+2,i+3}\to\left(\used^{d-1}_i 
           \wedge\used^{d-1}_{i+1}\right)\vee\left(\used^{d-1}_{i+2}
           \wedge\used^{d-1}_{i+3}\right)}{1\leq i\leq n-3} \]
This condition essentially forces the SAT solver to add a (redundant)
comparator on layer $d-1$, if both blocks have an unused channel in
that layer.

The same idea of having to add a comparator at layer $d-1$ is enforced
for the two remaining cases of a comparator followed by an unused
channel or its dual by $\psi_2^b$ and $\psi_2^c$, respectively:
\begin{align*}
  \psi_2^b &= \sset{g^d_{i,i+1}\wedge\neg\used^d_{i+2}\to\left(\used^{d-1}_i\wedge\used^{d-1}_{i+1}\right)\vee\used^{d-1}_{i+2}}{1\leq i\leq n-2}\\
  \psi_2^c &= \sset{\neg\used^d_i\wedge g^d_{i+1,i+2}\to\used^{d-1}_i\vee\left(\used^{d-1}_{i+1}\wedge\used^{d-1}_{i+2}\right)}{1\leq i\leq n-2}
\end{align*}

The final symmetry break is based on Lemma~\ref{lem:onetwo}, i.e., on
the idea of moving a comparator from the last layer to the second last
layer.
We encode that such a situation cannot occur, i.e., that whenever we
have a comparator on the last layer $d$ following or followed by an
unused channel, one of the channels of the comparator is used on layer
$d-1$:
\begin{align*}
  \psi_3^a &= \sset{g^d_{i,i+1}\wedge\neg\used^d_{i+2} \to\used^{d-1}_{i} \vee\used^{d-1}_{i+1}}{1\leq i\leq n-2}\\
  \psi_3^b &= \sset{g^d_{i,i+1}\wedge\neg\used^d_{i-1} \to\used^{d-1}_{i} \vee\used^{d-1}_{i+1}}{2\leq i\leq n-1}
\end{align*}

Defining $\psi_2=\psi_2^a\wedge\psi_2^b\wedge\psi_2^c$ and
$\psi_3=\psi_3^a\wedge\psi_3^b$, and denoting 
\[\mathit{last}=\bigwedge_{i=1}^4\varphi_i\wedge\bigwedge_{i=1}^3\psi_i\]
we can enrich the SAT encodings
expressed by Equations~(\ref{eq:sat}) and~(\ref{eq:satP})
to obtain 
\begin{equation}\label{eq:sat:last}
  \varphi^{\mathit{last}}(n,d) =\valid(C^d_n) \wedge \bigwedge_{\bar x \in \{0,1\}^n}\sorts(C^d_n, \bar x) \wedge\mathit{last}
\end{equation}
\begin{equation}\label{eq:satP:last}
  \varphi^{\mathit{last}}_P(n,d) =\valid(C^{d-|P|}_n) \wedge \bigwedge_{\bar x \in\outputs(P)}\sorts(C^{d-|P|}_n, \bar x)\wedge\mathit{last}
\end{equation}
The constraints $\varphi_1(\ell)$ do not seem to have a positive impact on the
performance of the SAT solvers, so we will not use them hereafter.

We note that the results from Section~\ref{sec:cosub} cannot be
applied in this scenario: the proof strategy relies on the fact that
there must exist a sorting network whose two first layers have a
particular form, and this is incompatible with the construction in
Corollary~\ref{cor:cosub}.

\section{Back to the Front: Permuting Filters}
\label{sec:permuting}

In this section, we come back to Lemma~\ref{lem:sorted} and to the
improved SAT encoding described in Section~\ref{sec:improved}.  
When considering the proposition of Equation~\eqref{eq:satP} for a
fixed prefix $P$, the SAT encoding that searches for an extension of
$P$ to a sorting network is optimized with respect to the leading
zeroes and trailing ones in  $\outputs(P)$.

It is a well-known fact that, if $P$ can be extended to a sorting
network of depth $d$, then so can any prefix $P'$ obtained by
permuting the channels in $P$ followed by a procedure called
untangling (which basically turns back comparators that were turned
``upside down'' in the permutation step -- see~\cite{Parberry91} or
Exercise 5.3.4.16 of~\cite{Knuth73}).
Although $P'$ is equally good as $P$ in terms of the number of its
(unsorted) outputs and in terms of its potential to extend to a
sorting network of depth $d$, different permutations of $P$ can lead
to smaller SAT encodings in light of the optimizations based on
leading zeroes and/or trailing ones.

\begin{wrapfigure}[5]{r}{0.25\textwidth}
  \vspace*{-1em}
  \begin{tabular}{cc}
    \begin{sortingnetwork}{6}{0.35}
      \nodeconnection{ {1,2}, {3,4}, {5,6}}
    \end{sortingnetwork}
    &
    \begin{sortingnetwork}{6}{0.35}
      \nodeconnection{ {1,6}}
      \nodeconnection{ {2,5}}
      \nodeconnection{ {3,4}}
    \end{sortingnetwork} \\
    $L^1_P$ & $L^1_{BZ}$
  \end{tabular}
\end{wrapfigure}
This behavior can already be observed in the first layer: Parberry
\cite{Parberry91} used the first layer $L^1_P$ consisting of
comparators of the form $(2i-1, 2i)$, for all $1 \leq i \leq
\left\lfloor \frac{n}{2} \right\rfloor$, whereas Bundala and
Z\'avodn\'y~\cite{BZ14} employed a first layer $L^1_{BZ}$, which is a permutation
thereof consisting of comparators $(i, n+1-i)$, for all $1 \leq i \leq
\left\lfloor \frac{n}{2} \right\rfloor$. See the figure on the right
for an example on $6$ channels.
We define the \emph{window-size} of a vector $\vec x$ to be its length
minus the number of its leading zeros and trailing ones. Formally, for
$\vec x = (x_1, x_2, \ldots, x_n )$,
\[ \mathit{ws}(\vec x) = 
        n- \max\sset{a}{x_1 = x_2 = \cdots = x_a = 0} - 
           \max\sset{b}{x_{n-b+1} = x_{n-b+2} = \cdots = x_n = 1}
\]
By the previous considerations, the number of channels for which the output on the channel is not immediately determined by the input to this channel,
 and which have thus to be explicitly considered in the SAT formula specifying the behavior of the network for an input vector from a set $S$,
is equal to the sum of the window-sizes of
the vectors in $S$.
Table~\ref{tab:numRows} shows the sum of window-sizes of
$\outputs(L^1_P)$ and $\outputs(L^1_{BZ})$ for $3 \leq n \leq 17$
channels: a first layer $L^1_{BZ}$ always leads to a smaller SAT
encoding than the corresponding first layer $L^1_P$.
\begin{table}[htb]
\small
\[\begin{array}{l|ccccccccccccccc}
\toprule
n  & 3 & 4 & 5 & 6 & 7 & 8 & 9 & 10 & 11 & 12 & 13 & 14 & 15 & 16 & 17 \\
\midrule
L^1_P & 5 & 12 & 44 & 84 & 233 & 408 & 1{,}016 & 1{,}704 & 4{,}013 & 6{,}564 & 14{,}948 & 24{,}060 & 53{,}585 & 85{,}296 & 186{,}992 \\
L^1_{BZ} & 4 & 10 & 36 & 72 & 196 & 358 & 876 & 1{,}524 & 3{,}532 & 5{,}962 & 13{,}380 & 22{,}128 & 48{,}628 & 79{,}246 & 171{,}612 \\
\bottomrule
\end{array}\]
\caption{Number of channels to consider in the encoding after the first layer.}
\label{tab:numRows}
\end{table}

We have observed that some permutations of a given prefix are better
for SAT encodings than others, and that new results on optimal-depth
sorting networks derive from fixing $2$-layer prefixes. So, we
introduce a new criterion on  prefixes: select a permutation of the
prefix that minimizes the sum of the window-sizes of its outputs.
This idea is further illustrated in the following example.

\begin{example}
  Consider the three representations $C_1, C_2,$ and $C_3$ of a
  particular $2$-layer filter in $R_6$, presented in
  Figure~\ref{fig:permuting6channels}.
\begin{figure}[h]
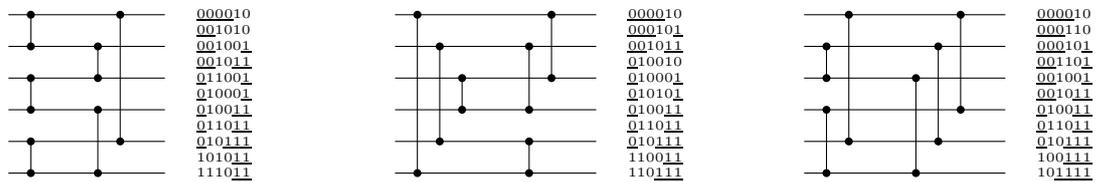

  \begin{minipage}{0.32 \textwidth}
    \centering
    \begin{sortingnetwork}{6}{0.6}
      \nodeconnection{{1,2},{3,4},{5,6}}
      \addtocounter{sncolumncounter}{2}
      \nodeconnection{{1,3},{4,5}}
      \nodeconnection{{2,6}}
     \end{sortingnetwork}
{\tiny$\begin{array}[b]{l}
\bu0 \bu0 \bu0 \bu0 1 0 \\
\bu0 \bu0 1 0 1 0 \\
\bu0 \bu0 1 0 0 \bu1 \\
\bu0 \bu0 1 0 \bu1 \bu1 \\
\bu0 1 1 0 0 \bu1 \\
\bu0 1 0 0 0 \bu1 \\
\bu0 1 0 0 \bu1 \bu1 \\
\bu0 1 1 0 \bu1 \bu1 \\
\bu0 1 0 \bu1 \bu1 \bu1 \\
1 0 1 0 \bu1 \bu1 \\
1 1 1 0 \bu1 \bu1 \\
\end{array}$}
    \end{minipage}
    \begin{minipage}{0.32 \textwidth}        
    \centering
     \begin{sortingnetwork}{6}{0.6}
        \nodeconnection{{1,6}}
        \nodeconnection{{2,5}}
        \nodeconnection{{3,4}}
        \addtocounter{sncolumncounter}{2}
        \nodeconnection{{1,2},{3,5}}
        \nodeconnection{{4,6}}
     \end{sortingnetwork}
{\tiny$\begin{array}[b]{l}
\bu0 \bu0 \bu0 \bu0 1 0 \\ 
\bu0 \bu0 \bu0 1 0 \bu1\\ 
\bu0 \bu0 1 0 \bu1 \bu1\\ 
\bu0 1 0 0 1 0\\ 
\bu0 1 0 0 0 \bu1\\ 
\bu0 1 0 1 0 \bu1\\ 
\bu0 1 0 0 \bu1 \bu1\\ 
\bu0 1 1 0 \bu1 \bu1\\ 
\bu0 1 0 \bu1 \bu1 \bu1\\ 
1 1 0 0 \bu1 \bu1\\ 
1 1 0 \bu1 \bu1 \bu1\\ 
\end{array}$}
    \end{minipage}
        \begin{minipage}{0.32 \textwidth}        
    \centering
     \begin{sortingnetwork}{6}{0.6}
        \nodeconnection{{1,3},{4,5}}
        \nodeconnection{{2,6}}
        \addtocounter{sncolumncounter}{2}
        \nodeconnection{{1,4}}
        \nodeconnection{{2,5}}
        \nodeconnection{{3,6}}
     \end{sortingnetwork}
{\tiny$\begin{array}[b]{l}
\bu0 \bu0 \bu0 \bu0 1 0\\ 
\bu0 \bu0 \bu0 1 1 0\\ 
\bu0 \bu0 \bu0 1 0 \bu1\\
\bu0 \bu0 1 1 0 \bu1\\ 
\bu0 \bu0 1 0 0 \bu1\\ 
\bu0 \bu0 1 0 \bu1 \bu1\\ 
\bu0 1 0 0 \bu1 \bu1\\ 
\bu0 1 1 0 \bu1 \bu1\\ 
\bu0 1 0 \bu1 \bu1 \bu1\\ 
1 0 0 \bu1 \bu1 \bu1\\ 
1 0 \bu1 \bu1 \bu1 \bu1\\ 
\end{array}$}
    \end{minipage}
    \caption{Three permutations of the same $2$-layer-prefix on $6$
      channels. $C_1$ (left) has $L^1_P$ as first layer, $C_2$
      (middle) has $L^1_{BZ}$ as first layer, and $C_3$ (right) is
      optimized to minimize total window size in the outputs from the
      filter. Each filter is accompanied by the set of its non-sorted
      outputs.}
    \label{fig:permuting6channels}
\end{figure}
Each of these prefixes can be transformed into the other two by
permuting channels and untangling, hence they all represent the same
filter.  
The outputs from the prefix $C_1$, on the left, give a total
window-size of $35$. The outputs from the prefix $C_2$, in the middle,
give a total window-size of $34$. The outputs from the prefix
$C_3$, on the right, give a total window-size of $28$.
\end{example}

To determine the existence of an $n$-channel sorting network of depth
$d$ we need to consider SAT instances for $|R_n|$ 2-layer prefixes.
For each prefix $P$, we seek a permutation, $P'$, that minimizes the
total window-sizes of the vectors in $\outputs(P')$.
To this end, we apply a simple evolutionary algorithm, randomly
swapping chosen channels followed by untangling to generate the
offspring. We use the sum of the window-sizes of the outputs as the
fitness function, and keep the $32$ best prefixes in each of  $20$
iterations. Finally we keep a single prefix with best fitness which we
denote by $F_\opt$. 

Table~\ref{tab:prefixStyles} shows SAT solving times when proving that
no sorting network on $16$ channels with $8$ layers exists using the
different representations of the 211 prefixes in $R_{16}$. Here,
$\FF_P$ denotes the representation of the $2$-layer prefixes such that
the first layer is $F_P$, and analogously for $\FF_{BZ}$ and
$\FF_{\opt}$. The evolutionary algorithm required 5 minutes to
compute the set $\FF_{\opt}$.  Optimizing the prefixes reduced the
overall SAT solving time by factors of circa $4$ and $2$, and the
maximum running times by factors of circa $9$ and $4$, respectively.

\begin{table}[htb]
\centering
   \begin{tabular}{l | r | r}
   \toprule
   Prefix & Overall time (s) & Maximum time (s) \\
   \midrule
   $\FF_P$ & $22{,}241$ & $326$ \\
   $\FF_{BZ}$ & $10{,}927$ & $150$ \\
   $\FF_\opt$ & $5{,}492$ & $36$ \\
   \bottomrule
   \end{tabular}
   \caption{Impact of permuting the prefix when proving that no sorting network for $16$ channels with at most $8$ layers exists.}
   \label{tab:prefixStyles}
\end{table}

The idea presented in this chapter -- to search for sorting networks
that extend a given prefix by first selecting a permutation of the
prefix to minimize total window size -- plays a central role in the
forthcoming chapters.

\section{Improved Upper Bounds for $17$, $19$ and 
                      $20$-Channel Sorting Networks}
\label{sec:upbound}

In this section, we show how the improvements presented in the previous
sections apply to facilitate the search for new upper bounds for
sorting networks with $17$, $19$ and $20$ channels. To do this, we exhibit
sorting networks for $17$ channels and for $20$ channels that have
one less layer than the previous best known sorting networks for these
cases -- $10$ and $11$ layers, respectively. The
$11$-layer sorting network for $20$ channels implies the improved
upper bound also for $19$ channels, which was previously equal to $12$.

When searching for upper bounds on the depth of sorting networks, it
is common practice to search for an extension of a particular chosen
prefix of the network. This is in contrast to the quest to prove
lower bounds, where we need to show that there is no prefix that
extends to a network of the specified depth. For upper bounds, any
choice of prefix is fair play. 

For $n=17$ channels, we chose a prefix consisting of the first three
layers of a Green Filter on $16$ channels, leaving one channel unused,
and we were able to extend it to the sorting network with $10$ layers
depicted as Figure~\ref{snw:17_10}.
The selected prefix, separated by the dashed line, produces $800$
outputs, including the $18$ sorted ones.
This network improves on the best previously known sorting network for
$17$ channels.
For $n=20$ channels, we chose a $4$-layer prefix constructed from a
$4$-layer Green filter on $16$ channels (the top $8$ and bottom $8$),
and where the remaining middle $4$ channels are connected with a
$3$-layer optimal-depth sorting network.
We were able to extend this prefix by an
additional $7$ layers to obtain the $11$-layer sorting network
depicted as Figure~\ref{fig:20_11}, where the selected prefix is again
marked by the dashed line. This improves on the best previously known
sorting network for $20$ channels, which consisted of $12$ layers. The
selected prefix produces $840$ outputs, including $21$ sorted ones.
As reported in a preliminary
paper~\cite{DBLP:conf/cie/EhlersM15}, both of these results can be found
using a straightforward SAT encoding.

\begin{figure}
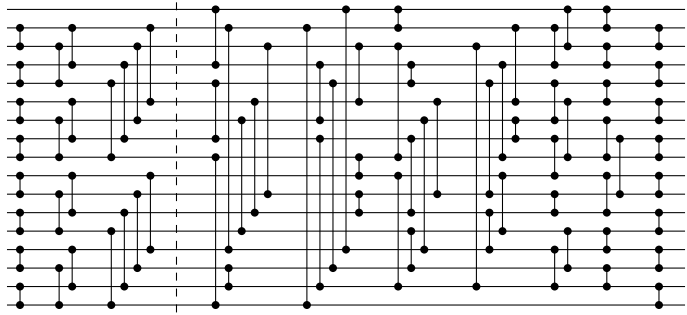

\centering
\begin{sortingnetwork}{17}{0.35}
      \nodeconnection{ {1,2}, {3,4}, {5,6}, {7,8}, {9,10}, {11,12}, {13,14}, {15,16}}
        \addtocounter{sncolumncounter}{2}
        \nodeconnection{ {1,3}, {5,7}, {9,11}, {13,15}}
        \nodeconnection{ {2,4}, {6,8}, {10,12}, {14,16}}
        \addtocounter{sncolumncounter}{2}      
        \nodeconnection{ {1,5}, {9,13}}
        \nodeconnection{ {2,6}, {10,14}}
        \nodeconnection{ {3,7}, {11,15}}
        \nodeconnection{ {4,8}, {12,16}}
      \nextlayer
        \addtocounter{sncolumncounter}{2}
        \nodeconnection{ {1,9}, {10,13}, {14,17}}
        \nodeconnection{ {2,3}, {4,16}}
        \nodeconnection{ {5,11}}
        \nodeconnection{ {6,12}}
        \nodeconnection{ {7,15}}
        \addtocounter{sncolumncounter}{2}
        \nodeconnection{ {1,16}}
        \nodeconnection{ {2,10}, {11,14}}
        \nodeconnection{ {3,13}}
        \nodeconnection{ {4,17}}
        \nodeconnection{ {6,7}, {8,9}, {12,15}}
        \addtocounter{sncolumncounter}{2}
        \nodeconnection{ {2,8}, {9,15}, {16,17}}
        \nodeconnection{ {3,5}, {6,10}, {13,14}}
        \nodeconnection{ {4,11}}
        \nodeconnection{ {7,12}}
        \addtocounter{sncolumncounter}{2}
        \nodeconnection{ {2,15}}
        \nodeconnection{ {4,6}, {7,13}}
        \nodeconnection{ {5,8}, {9,14}}
        \nodeconnection{ {10,11}, {12,16}}
        \addtocounter{sncolumncounter}{2}
        \nodeconnection{ {2,4}, {6,7}, {8,10}, {11,13}, {14,16}}
        \nodeconnection{ {3,5}, {9,12}, {15,17}}
        \addtocounter{sncolumncounter}{2}
        \nodeconnection{ {2,3}, {4,5}, {6,8}, {9,11}, {12,13}, {14,15}, {16,17}}
        \nodeconnection{ {7,10}}
        \addtocounter{sncolumncounter}{2}
        \nodeconnection{ {1,2}, {3,4}, {5,6}, {7,8}, {9,10}, {11,12}, {13,14}, {15,16}}
    \end{sortingnetwork}
    \caption{A $10$-layer sorting network on $17$ channels with a $3$-layer Green prefix.}
\label{snw:17_10}\end{figure}
\begin{figure}
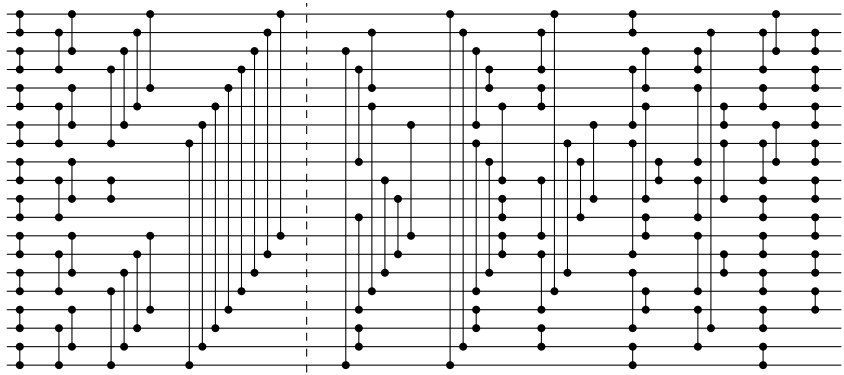

    \centering
    \begin{sortingnetwork}{20}{0.35}
        \nodeconnection{ {1,2}, {3,4}, {5,6}, {7,8}, {9,10}, {11,12}, {13,14}, {15,16}, {17,18}, {19,20}}
        \addtocounter{sncolumncounter}{2}
        \nodeconnection{ {1,3}, {5,7}, {9,11}, {13,15}, {17,19}}
        \nodeconnection{ {2,4}, {6,8}, {10,12}, {14,16}, {18,20}}
        \addtocounter{sncolumncounter}{2}
        \nodeconnection{ {1,5}, {10,11}, {13,17}}
        \nodeconnection{ {2,6}, {14,18}}
        \nodeconnection{ {3,7}, {15,19}}
        \nodeconnection{ {4,8}, {16,20}}
        \addtocounter{sncolumncounter}{2}
        \nodeconnection{ {1,13}}
        \nodeconnection{ {2,14}}
        \nodeconnection{ {3,15}}
        \nodeconnection{ {4,16}}
        \nodeconnection{ {5,17}}
        \nodeconnection{ {6,18}}
        \nodeconnection{ {7,19}}
        \nodeconnection{ {8,20}}  \nextlayer
        \addtocounter{sncolumncounter}{2}
        \nodeconnection{ {1,18}}
        \nodeconnection{ {2,3}, {4,9}, {12,17}}
        \nodeconnection{ {5,15}, {16,19}}
        \nodeconnection{ {6,11}}
        \nodeconnection{ {7,10}}
        \nodeconnection{ {8,14}}
        \addtocounter{sncolumncounter}{2}
        \nodeconnection{ {1,20}}
        \nodeconnection{ {2,19}}
        \nodeconnection{ {3,4}, {5,13}, {14,18}}
        \nodeconnection{ {6,12}, {16,17}}
        \nodeconnection{ {7,8}, {9,10}, {11,15}}
        \addtocounter{sncolumncounter}{2}
        \nodeconnection{ {2,3}, {4,7}, {8,11}, {15,16}, {17,19}}
        \nodeconnection{ {5,20}}
        \nodeconnection{ {6,13}}
        \nodeconnection{ {9,12}}
        \nodeconnection{ {10,14}}
        \addtocounter{sncolumncounter}{2}
        \nodeconnection{ {1,2}, {3,6}, {7,13}, {14,17}, {19,20}}
        \nodeconnection{ {4,5}, {8,9}, {10,15}, {16,18}}
        \nodeconnection{ {11,12}}
        \addtocounter{sncolumncounter}{2}
        \nodeconnection{ {2,4}, {5,8}, {9,11}, {12,16}, {17,18}}
        \nodeconnection{ {3,19}}
        \nodeconnection{ {6,7}, {10,13}, {14,15}}
        \addtocounter{sncolumncounter}{2}
        \nodeconnection{ {1,2}, {3,4}, {5,6}, {7,8}, {9,10}, {11,13}, {15,16}, {17,19}}
        \nodeconnection{ {12,14}, {18,20}}
        \addtocounter{sncolumncounter}{2}
        \nodeconnection{ {4,5}, {6,7}, {8,9}, {10,11}, {12,13}, {14,15}, {16,17}, {18,19}}
        
    \end{sortingnetwork}
    \caption{An $11$-layer sorting network on $20$ channels with the
      selected $4$-layer prefix.}
    \label{fig:20_11}
\end{figure}

The techniques presented in this paper massively improve the
computational time required to find the depth-$10$ and depth-$11$
sorting networks for $17$ and $20$ channels,
respectively. Table~\ref{tab:results_17} presents the impact of the
three main contributions: introducing the constraints on the last
layers (column one), applying the improved SAT encoding (column two),
and optimizing the selected prefix (column three).
The bottom row of the table specifies that, without applying the
techniques of this paper, the SAT solving times to find the two
networks are $9{,}995$ and $45{,}596$ seconds, respectively.
The first row of the table specifies that, when applying all three
techniques together, the SAT solving times to find the two networks are $17$ and
$46$ seconds, respectively. This constitutes speed-ups by a factor of $587$
in the case of $17$ channels, and $991$, for $20$ channels.

Closer examination of the table highlights the impact of each
combination of the three techniques. For example, the third row from
the bottom in comparison to the bottom row, indicates speed-ups in
order(s) of magnitude in the search for a depth-$10$ and depth-$11$
sorting networks on $17$ and $20$ channels when introducing only the improved
SAT encoding described in Section~\ref{sec:improved}.
The actual networks found with each configuration of the solver might
differ; for example, the networks in Figures~\ref{snw:17_10}
and~\ref{fig:20_11}, which were found using the improved encoding, do
not have an optimized prefix; that in Figure~\ref{fig:20_11} also does
not have a co-saturated last layer. If these optimizations are also used,
different sorting networks of the same depth will be found.

\begin{table}[htdp]
\begin{center}
\begin{tabular}{c | c | c | r | r}
    \toprule
 Symmetry Break & Improved Encoding & Opt. Prefix & \multicolumn{2}{c}{ Time/s} \\
 (last layers)	& 		    & 		  & $n=17$ & $n=20$ \\
                \midrule
 yes & yes & yes & $17$ & $46$ \\
 yes & yes & no & $78$ & $560$ \\
 yes & no & yes & $37$ & $13{,}199$ \\
 yes & no & no & $265$ & $1{,}255$\\
 no & yes & yes & $64$ & $97$\\
 no & yes & no & $838$ & $412$ \\
 no & no & yes & $3{,}723$ & $36{,}159$ \\
 no & no & no & $9{,}995$ & $45{,}596$ \\
 \bottomrule
\end{tabular}
\end{center}
\caption{Impact of the different optimizations in the time required to find the new sorting networks on $17$ and $20$ channels.}
\label{tab:results_17}
\end{table}%

\section{Depth $10$ is Optimal for Sorting $17$ Inputs}
\label{sec:lowbound}
After having found a sorting network for $17$ inputs using only $10$
layers, we seek to prove its optimality. This task can be decomposed
in showing that none of the $609$ prefixes in $R_{17}$ can be extended
to a $9$-layer sorting network~\cite{BZ14}.  

In a first attempt, we tried to prove this by using a straightforward
SAT encoding, basically identical to the one introduced
in~\cite{BZ14}, enriched by the symmetry-breaking constraints on the
last layers described in Section~\ref{sec:last}. Using $150$ CPU
cores, we started a SAT solver for each of the first $150$
cases. Whenever a filter was shown unfeasible, we started a solver for
the next case which had not yet been considered. Before we broke up
this experiment after $48$ days, we were able to prove that $381$ out
of $609$ prefixes cannot be extended to sorting networks of depth $9$.
Showing unsatisfiability of these formulae took a total of $353 \cdot
10^6$ seconds of CPU time, with a maximum of $3 \cdot 10^6$ seconds.

In a second attempt, we applied all three of the optimizations
described in this paper: introducing the constraints on the last
layers, applying the improved SAT encoding, and optimizing the
selected prefix. For the latter, for each of the $609$ distinct $2$-layer prefixes
in $R_{17}$ we chose a permutation minimizing the size of the
windows in the outputs of the prefix as explained in
Section~\ref{sec:permuting}.  This time, we were able to prove that
none of the $609$ prefixes can be extend to a $9$-layer sorting network.
The overall CPU time for all $609$ instances was $27.63 \cdot 10^6$
seconds, with a maximum running time of $97{,}112$ seconds. Optimizing
the prefixes took a total of 25 minutes.  This is a speed up of factor
at least $42.7$ concerning the maximum running time, and $20.4$ for
the average running time.  Since the result for all $2$-layer prefixes
was unsat, we conclude:
\begin{theorem}%
Every sorting network for $n \geq 17$ channels has at least $10$ layers.
\end{theorem}

\section{Conclusions}
\label{sec:concl}

The contributions of this paper can be divided in theoretical insights into, improved methods for, and consequent new results on depth-optimal sorting networks.

Based on the first systematic exploration of what happens at the \emph{end} of a sorting network, we have discovered both necessary and sufficient properties of the last layers of depth-optimal sorting networks. To this end we have introduced the general concept of $k$-blocks, allowing to reason about the sorting behavior of suffixes of sorting networks. Specifically, we have shown that non-redundant comparators on the last layer must connect adjacent channels, and that there is always a depth-minimal sorting networks with $L^1_P$ as the last layer. We also discovered a sufficient condition for when adding a comparator to a sorting network results in another sorting network.

We have improved SAT-based methods for deciding whether a prefix can be extended to a sorting network of a given depth in three ways. First, based on the novel theoretical insights, we have added symmetry breaking constraints regarding the last layers of the network. Second, we have optimized the encoding for windowed inputs by trading off a slight increase in variables for vast reductions of the number of clauses and literals. Third, we have shown that not all prefixes are created equal and that, indeed, choosing a permutation of the prefix that maximizes leading zeros and trailing ones leads to a significant reduction of the overall size of the SAT problem. Computer experiments have shown that all three improvements separately translate to significant improvements in SAT solving time, yielding several orders of magnitude when used synergistically.

Exploiting the improved SAT-based methods, we have been able to find sorting networks of lesser depth than previously known for $17$, $20$, and (consequently) $19$ channels, lowering the upper bounds for optimal depth to $10$, $11$, and $11$, respectively. In addition, we have proved optimality of our sorting network on $17$ channels using extensive computer experiments, determining the optimal depth for this case to be $10$, and pushing the boundary of our knowledge one step further.
%

%
%
%
%

%



\end{document}